\documentclass[12pt]{article}
\usepackage{graphicx} 
\usepackage{epsfig}
\usepackage{amsmath, amsfonts, amsthm,amssymb}
\usepackage{latexsym}
\usepackage{xcolor}
\usepackage{mathrsfs}
\usepackage{natbib}
\usepackage{hyperref}
\usepackage{multirow}
\usepackage{appendix}
\usepackage{dsfont}
\usepackage[left=3cm,top=4cm,right=3cm,bottom=3cm]{geometry}
\RequirePackage[caption=false]{subfig}
\usepackage{lscape}
\usepackage{dsfont}
\usepackage{rotating}
\usepackage{ulem}
\newtheorem{prop}[equation]{Proposition}

\usepackage{mathabx}
\usepackage{calligra}

\usepackage{natbib}
\usepackage{amsmath}
\usepackage{dsfont}
\usepackage{subfig}
\usepackage{xcolor}
\usepackage{ulem}
\usepackage{amsfonts}
\usepackage{booktabs}
\usepackage{float}

\newcommand{\ee}{\mathds{E}}
\newcommand{\1}{\mathds{1}}
\DeclareMathOperator{\card}{card}
\DeclareMathOperator{\de}{\mathrm{d}}
\DeclareMathOperator{\tf}{\mathfrak{t}}
 
\usepackage{calligra}

\newcommand{\hypref}[1]{\hyperref[#1]{H\ref{#1}}}
\usepackage{tikz}
\usetikzlibrary{arrows.meta, positioning}

\begin{document}

\title{\LARGE \bf{Testing the Structural Properties of Marked Point Processes Using Local Inhomogeneous Mark-Weighted K-Functions} }
\maketitle
\begin{center}
{{\bf Nicoletta D'Angelo$^{1}$},  {\bf Giada Adelfio$^{1}$} and {\bf Matthias Eckardt$^{2}$}}\\
\noindent $^{\text{1}}$ Department of Economics, Business and Statistics, University of Palermo, Palermo, Italy\\
\noindent $^{\text{2}}$ Chair of Statistics, Humboldt-Universit\"{a}t zu Berlin , Berlin, Germany 
\end{center} 
\begin{abstract}
This work proposes $\chi^2$-type test statistics to assess different hypotheses on the local structure of an observed marked point pattern. The test statistics is based on the local inhomogeneous extension of the mark-weighted $K$-function to investigate local behaviour of the marked point pattern. The summary statistic captures interactions between marks and locations by assessing local contributions to global deviations from independence or homogeneity. The methodology proves to be effective in identifying both global and localised departures from the null hypotheses, even in scenarios with subtle mark structures or small sample sizes. Real-world environmental applications to forestry and earthquake data demonstrate the utility of the proposed framework for detecting spatially dependent marked structures in the patterns.

\end{abstract}
{\it Keywords:  local characteristics, spatial statistics, point processes, $K$-function, hypothesis testing}

\section{Introduction}

In the analysis of marked spatial point processes, a fundamental objective involves selecting a suitable point process model, enabling precise estimation and forecasting of spatial phenomena derived from observed point patterns. The academic discourse acknowledges the critical role of combining marks, points, and their integrative analysis.

However, the complexities of their interactions remain inadequately explored and necessitate comprehensive evaluation through summary characteristics. Typically, these models presuppose that marks (observations) and their spatial coordinates (points) are mutually independent. This assumption was scrutinized by \cite{guan2006tests}, who devised innovative graphical and formal testing techniques. Previously, \cite{schoenberg2004testing} concentrated on the nonparametric evaluation of separability in spatial-temporal marked point processes, illustrating the efficacy of Cramer–von Mises-type tests in identifying minor deviations from separability. When these techniques were applied to wildfire data from Los Angeles County, they revealed fire clustering by similar sizes over periods extending up to 3.9 years, thereby casting doubt on the appropriateness of the separability hypothesis. 

Understanding the interplay between marks and points is crucial, as it guides the selection of point process models that incorporate specific assumptions regarding the distribution of marks relative to their points' locations. In scenarios where marks are shown to be independent of point locations, an unmarked point process model is appropriate. In contrast, where evidence suggests mark-location dependency, it becomes imperative to consider a marked point process model \citep{tarantino2025modeling}. The degree of separability between marks and points informs the exploration of various models. Global summary statistics, such as the $K$-function, serve as prevalent tools in characterizing spatial interactions within point processes. Yet, their components, known as Local Indicators of Spatial Association (LISA) \citep{anselin:95}, possess the capability to reveal local structures and outliers. Initially designed for areal datasets (for instance, see \cite{anselin1996chapter,getis:ord:92}), these local statistics have been adapted for application in spatial point processes, as illustrated by the local $K$-function \citep{getis:franklin:87}. Progress in this area has led to the development of the local product density \citep{mateu:lorenzo:porcu:10}, with subsequent application in areas such as image analysis \citep{mateu:lorenzo:porcu:07} and epidemiology \citep{moraga:montes:11}. This concept has further evolved into the spatio-temporal domain through LISTA functions \citep{siino2018testing}, facilitating analyses of clustering on Euclidean domains \citep{adelfio2020some} and linear networks \citep{dangelo2021assessing}, as well as model fitting \citep{dangelo2021locally,d2024constructed}. 

Within marked point processes, these local statistics provide insights that complement those yielded by global summary statistics. \cite{d2024local} have formalised a new class of summary statistics called \textit{local $t$-weighted marked $n$-th order inhomogeneous $K$-functions}, extending beyond the second order for general marked point processes. Recently, \citet{LIMA} introduced \textit{Local Indicators of Mark Association} (LIMA) which unify various local mark correlation functions for real- and object-valued marks in a general framework. The type of local dependence structures these summary statistics encapsulate is dictated by the choice of the weight, the so-called \textit{test function}, applied. \cite{d2024local}  established fundamental properties of these functions, demonstrating how these statistical tools could generate a broad array of existing summary statistics for marked point processes. 
Subsequently, the local $t$-weighted marked $n$-th order inhomogeneous $K$-functions were utilized to develop a novel local test intended to assess random labeling, focusing on evaluating potential dependencies of functional marks based on spatial point distribution, with a permutation testing framework. 

In our research, we devised a novel statistical testing procedure based on the mark-weighted $K$ function by adopting a $\chi^2$-test formulation. Using an empirical distribution of the proposed method under independence assumption, the proposed test is intended to provide detailed guidance for formally evaluating critical structures within the point patterns, mark distributions, and their collective aspects. Notably, this methodology encompasses both the traditional mark-weighted $K$-function and extends further to identify local structures using the local mark-weighted $K$-function. The reasoning behind employing a formal test statistic derived from the $K$-function instead of solely relying on summary attributes of the marks is rooted in its ability to address potential inhomogeneities in the foundational point pattern. The efficacy of this test is evaluated through simulation studies, and two applications are examined to cover the possible scenarios. These analyses substantiated that the test capably detects mark-dependence structures in marked point patterns, showcasing improvements as both the number of points and the intensity of mark-dependence increase.
All the analyses are performed through the statistical software R \citep{R} and are available from the authors.

The organization of this paper can be delineated in the following manner. Section \ref{sec:prel} introduces essential preliminaries concerning marked point processes and their associated summary statistics, particularly focusing on the \textit{locally mark-weighted $K$-function}. In Section \ref{sec:test}, we introduce and describe the proposed test, offering both its global and local variants. Section \ref{sec:sims} is dedicated to simulations to assess the proposed test methodologies. In Section \ref{sec:appls}, we explore two environmental applications, illustrating the potential scenarios. Concluding remarks and insights are thoroughly discussed in Section \ref{sec:concl}, synthesizing the findings and implications of our study.

\section{Preliminaries on marked point processes}\label{sec:prel}

Throughout, let $X=\lbrace(x_i,m(x_i))\rbrace_{i=1}^N$ denote a marked spatial point process on $\mathds{R}^d \times \mathds{M}$ where $\mathds{M}$, the so-called mark space, is assumed to be a complete separable metric space equipped with a finite reference measure $\nu$ on the Borel $\sigma$-algebra $\mathcal{B(\mathds{M})}$. 
Let $\Breve{X}$ denote the ground, i.e., unmarked point process associated Lebesgue measure $\vert A \vert =\int_A\de z$ for Borel sets $A\in\mathcal{B}(\mathds{R}^d)$. Further, we write $b(x,r)$ to denote a closed Euclidean $r$-ball around $x\in\mathds{R}^d$. The Borel $\sigma$-algebra  $\mathcal{B}(\mathds{R}^d \times  \mathds{M})=\mathcal{B}(\mathds{R}^d)\otimes\mathcal{B}(\mathds{M})$ is endowed with the product measure $A\times E\mapsto \vert A \vert \nu(E)$, $A\times E\in\mathcal{B}(\mathds{R}^d \times  \mathcal{M})$, with $\nu(x)$ the Lesbegue measure of $x$.  
Note that, formally, $Y$ can be considered to be a random element in the measurable space $(N_{lf},\mathcal{N})$ of locally finite point configurations/patterns $\mathbf{x}=\lbrace(x_1,m(x_1)),\ldots, (x_n,m(x_n))\rbrace$, $n\geq0$ \citep{daley:vere-jones:08,van2000markov}. For any sets $A\in\mathcal{B}(\mathds{R}^d)$ and $E\in\mathcal{B}(\mathds{M})$, the cardinality of the random set $Y\cap(A\times E)$ is  $Y(A\times E)=\sum_{(x,m(x))\in Y}\1\lbrace(x,m(x))\in A\times E\rbrace$, where $\1$ denotes an indicator function. In general, we assume $X$ to be simple, implying that $Y$ does not contain any coincident points. 

The expected number of points in $A\in \mathcal{B}(\mathds{R}^d)$ and marks in $E \in \mathcal{B}(\mathds{M})$ is
\begin{equation*}
    \Lambda^m(A\times E) = \mathds{E}[\card(X \cap (A\times E))] = \int_{A\times E} \lambda^m((u,m(u))\de u\,\nu(\de m(u)),
\end{equation*}
where $\lambda^m(\cdot)$ is the intensity function of $X$, and governs its spatial distribution. Likewise, the intensity function of the ground process $\Breve{X}$ will be denoted by $\lambda(\cdot)$. 

\subsection{Global and local mark correlation functions}
For $X$ the association of the marks is usually investigated through second-order summary characteristics including so-called \textit{cross}/\textit{dot}-type and mark-weighted versions and various mark correlation functions for categorical and real/object-valued marks, respectively   
\citep[see][for detailed treatment]{Eckardt:Moradi:currrent,Eckardt2024Rejoinder}. In the later case, classic, i.e., global mark summary statistics display the average association/variation among marks as a function of an interpoint distance $r \geq 0$, aiming at uncovering the average space-dependent distributional behaviours for marks. To this end, a test function $\tf:\mathds{M}\times\mathds{M}\mapsto\mathds{R}$ is employed to construct the specific global mark correlation functions in their most general form. Precisely, given the marks at any pair of point at interpoint distance $r$, the global mark correlation functions are constructed by taking the conditional expectation of the test function,
\begin{equation*}
    c_{\tf}(r) = \mathds{E} \left[ \tf \left( m(x), m(y) \right) \mid (x, m(x)), (y, m(y)) \in X \right].
\end{equation*}
Under the assumption of mark independence, i.e., when \(r \to \infty\), we have
\begin{equation*}
    c_{\tf}(\infty) = \int_{\mathds{R}} \int_{\mathds{R}} \tf \left( m(x), m(y) \right) \nu(dm(x)) \nu(dm(y)).
\end{equation*}
Normalising \(c_{\tf}(r)\) by  the conditional expectation of the test function under independent mark \(c_{\tf} = c_{\tf}(\infty)\) gives rise to the so-called \(\tf\)-\textit{correlation function} \(\kappa_{\tf}(r)\) as
\begin{equation*}
    \kappa_{\tf}(r) = \frac{c_{\tf}(r)}{c_{\tf}},    
\end{equation*}
whose precise form and interpretation depends on the specification of the test function under study. By contruction, \(c_{\tf}(r)\) coincides with \(c_{\tf}\) if there is no structure in the marks, i.e. the i.i.d mark assumption. Under the assumption of mark independence, \(c_{\tf}(r)\) coincides with \(c_{\tf}\). Whence, \(\kappa_{\tf}(r)\) equals one, and, therefore, any deviations from unity indicate the existence of associations/variations among marks.
 Note that, due to the distinct
forms of these mark correlation functions, they each have their normalising factors.
For instance, in the case of Stoyan’s mark correlation function with test function $\tf(m(x),m(y))=m(x)m(y)$, we have
\begin{equation*}
 c_{\tf}(\infty) = \int_{\mathds{R}} \int_{\mathds{R}} \tf(m(x),m(y)) \nu (\de m(x)) \nu(\de m(y))=\int_{\mathds{R}} \int_{\mathds{R}} m(x) m(y) \nu (\de m(x)) \nu(\de m(y))=\mu^2  
\end{equation*}
with $\mu^2$ denoting the mean mark squared.

Instead of a global formulation, \cite{LIMA} introduced local text functions with take the mark at a fixed point $x_i$ and the mark of any alternative point as argument.  In particular, they defined an local $\tf$ correlation function 
\begin{eqnarray*}
\kappa_{\tf,i}(r)
=
\frac{
c_{\tf,i}(r)
}{
c_{\tf,i}
}. 
\end{eqnarray*}
with $c_{\tf,i}(r)=\mathds{E}_{(x_i, m(x_i))} 
 \left[
 \tf_{i} (m(x_i),m(y))
\Bigl\vert  
  (y, m(y)) \in X\setminus (x_i, m(x_i)),\ d(x_i, y) =r 
 \right]$ and $c_{\tf,i}=c_{\tf,i}(\infty)$.

\subsection{Global and local mark-weighted $K$-function}

Different from the above construction principle of mark correlation functions, \cite{pettinen1992forest} proposed to include the test functions as a weight into the $K$-function yielding a positive or negative shift of the empirical curve in the case that the mark independence assumption is not fulfilled. In detail, its mark-weighted 
$K$-function takes the form   
\begin{equation}\label{eq:Kmm}
 K_{\tf}(r)
 =
 \frac{1}{\lambda c_{\tf}}
 \ee
 \left[
 \sum_{\substack{x \in X \\ y \in X,\, y \neq x}} 
 \tf \left(
    m(x),m(y)
    \right)
 \1\lbrace d(x,y)\leq r \rbrace
 \right]
\end{equation}
with $\tf$ needed to be defined by the user, given the data under study.  The normalisation is designed so that, under random labelling, $K_{\tf} (r) $ is equal to the Ripley's $K$-function \citep{ripley:76}. Focusing on the general interpretability of $K_{{\tf}}(r)$ we here explicitly restrict on test function 
\begin{equation*}
\tf(m(x),m(y))=m(x) m(y)
\end{equation*} 
with normalising term $c_{\tf}=\mu^{2}$ as applied in Stoyan's mark correlation function. If the product of the marks is on average larger than the expected case under the null at a certain distance $r$, the $K$ function will be positively shifted at that distance, allowing to identify spatially dependent marks and vice versa.  

Noting that \eqref{eq:Kmm} requires stationarity and isotropic assumptions, \cite{d2024local} proposes an extension of the mark-weighted $K$-function ($K_{\tf}$) function  to the inhomogeneous case where the intensity of the process is allowed to vary between the distinct point locations. Precisely, they defined the inhomogeneous $ K_{{\tf}}^{\mathrm{inhom}}$ by 
\begin{equation}\label{eq:InhKmm}
 K_{\tf}^{\mathrm{inhom}}(r)
 =
 \frac{1}{c_{\tf}}
 \ee
 \left[
 \sum_{\substack{x \in X \\ y \in X,\, y \neq x}} 
 \frac{1}{\lambda(x)}
 \tf \left(
    m(x),m(y)
    \right)
 \1\lbrace d(x,y)\leq r \rbrace
 \right].
\end{equation}
Changing the test function into a local version  and rewriting \eqref{eq:InhKmm}  yields 
\begin{equation}\label{eq:InhKmmLocal}
 K_{\tf,i}^{\mathrm{inhom}}(r)
 =
 \frac{1}{c_{\tf,i}}
 \ee
 \left[
 \sum_{\substack{ y \in X\\ y \neq x_i}} 
 \frac{1}{\lambda(x_i)}
 \tf_i(m(x_i)m(y))
 \1\lbrace d(x,y)\leq r \rbrace
 \right]
\end{equation}
with $\tf_i(m(x_i)m(y))=m(x_i)m(y)$, where the location $x_i$ is considered to be fixed as defined in \cite{LIMA} and $c_{\tf,i}=m(x_i)\int m(y)\nu(\de m(y))=m(x_i)\mu_y$. 

In the following, we denote by $K_{{\tf}}^{\mathrm{inhom}}(r)$ the inhomogeneous version of the mark-weighted $K$-function.

A possible estimator of \eqref{eq:InhKmm} is
\begin{equation}\label{eq:InhKmmEst}
\widehat{K}_{\tf}^{\mathrm{inhom}}(r)
=
\frac{1}{c_{\tf}}
\sum_{x\in X}\sum_{\substack{y\in X\\ y\neq x}}
\frac{\tf\big(m(x),m(y)\big)\;\mathbf{1}\{d(x,y)\le r\}}
{\widehat{\lambda}(x)\,\widehat{\lambda}(y)} ,
\end{equation}
and its local version can be obtained through the individual contribution of each point to the global $K$-function, as follows 

\begin{equation}\label{eq:prev}
\widehat{K}_{\tf,i}^{\mathrm{inhom}}(r)
=
\frac{1}{c_{\tf,i}}
\sum_{\substack{y\in X\\ y\neq x_i}}
\frac{\tf_i\big(m(x_i),m(y)\big)\;\mathbf{1}\{d(x_i,y)\le r\}}
{\widehat{\lambda}(x_i)\,\widehat{\lambda}(y)} ,
\qquad i=1,\dots,n.
\end{equation}

Finallt, the \textit{homogeneous} local estimator is obtained by imputing a constant intensity in \eqref{eq:prev} as follows

\begin{equation}\label{eq:K}
\hat{K}_{\tf,i}(r)
 =
 \frac{1}{c_{\tf,i}\hat{\lambda}^2} \sum_{x_i \in X}\sum_{y \neq x_i} \tf_i \left(
    m(x_i),m(y)
    \right)
 \1\lbrace d(x_i,y)\leq r \rbrace.
\end{equation}

\section{A test statistics for marked point process hypothesis testing} \label{sec:test}

This section is devoted to the definition of the proposed statistical tests for marked point process hypothesis testing. In detail, Section \ref{sec:global_test} introduces three systems of hypothesis for assessing the global structure of the marked point patterns, while Section \ref{sec:local_test} extends these to the local setting by the definition of local hypotheses and related statistical tests.

\subsection{Global test}\label{sec:global_test}

We aim to build a test statistic allowing us to assess the following hypothesis system:

\begin{enumerate}
\item[]
\begin{equation*}
\begin{aligned}
\mathrm{H}1 \; &= 
\begin{cases}
\mathcal{H}_{0}: & \text{spatial homogeneity with independent marks} \\
\mathcal{H}_{1}: & \text{deviation from spatial homogeneity and/or mark independence} 
\end{cases} 
\end{aligned}
\end{equation*}
\end{enumerate}

Then, depending on the result of $\mathrm{H}1$, one may proceed with a sequential hypothesis testing as follows. 

\begin{center}
\resizebox{0.9\textwidth}{!}{
\begin{tikzpicture}[
node distance=1.6cm and 2cm,
every node/.style={font=\normalsize},
test/.style={rectangle, draw, rounded corners, align=center, minimum width=2.2cm, minimum height=0.8cm},
result/.style={rectangle, draw, align=center, minimum width=2.8cm, minimum height=0.8cm},
arrow/.style={-Latex}
]

\node[test] (H1) {Test $\mathrm{H}1$};

\node[result, below left=of H1, xshift=-0.5cm] (R1) {Homogeneous points \\ with independent marks};

\node[coordinate, below=of H1] (Bbot) {};

\node[test, below left=of Bbot] (H2) {Test $\mathrm{H}2$};
\node[test, below right=of Bbot] (H3) {Test $\mathrm{H}3$};

\node[result, below left=of H2] (R2) {Inhomogeneous points};
\node[result, below right=of H2] (R3) {Inhomogeneous points \\ with dependent marks};
\node[result, below right=of H3] (R5) {Dependent marks};

\draw[arrow] (H1) -- node[left]{$\mathcal{H}_0$} (R1);
\draw[arrow] (H1) -- node[right]{$\mathcal{H}_1$} (Bbot);

\draw[arrow] (Bbot) -- (H2);
\draw[arrow] (Bbot) -- (H3);

\draw[arrow] (H2) -- node[left]{$\mathcal{H}_0$} (R2);
\draw[arrow] (H2) -- node[right]{$\mathcal{H}_1$} (R3);

\draw[arrow] (H3) -- node[right]{$\mathcal{H}_0$} (R5);
\draw[arrow] (H3) -- node[left]{$\mathcal{H}_1$} (R3);

\end{tikzpicture}
}
\end{center}
The global hypothesis $\mathrm{H}1$ is used to detect the presence of any spatial structure in the marked point pattern. If $\mathrm{H}1$ is rejected, hypotheses $\mathrm{H}2$ and $\mathrm{H}3$ are subsequently tested in order to determine whether the detected structure arises from the marks, the spatial distribution of the points, or both.

Each test employs an appropriate statistic to capture deviations from the corresponding null hypothesis. The outcomes of the sequential tests identify one of the four possible global configurations.

See Figure \ref{fig:enter-label} for an overview of the possible cases, using patterns with an expected number of $100$ points in the unit square. In the first column,  point patterns are represented, either inhomogeneous (first and second rows), or homogeneous (third and fourth rows). The second column displays the local mark-weighted $K$-functions in light blue and the global mark-weighted $K$-function in blue. The \textit{unmarked} versions of such a $K$-function is denoted by $K(r)$ and it is displayed in black. The expectation of $K(r)$ under the assumption of complete spatial randomness and non-spatially dependent marks is in red.  The third column shows the mark-weighted $K$-functions (in blue), $K$-function (in black), and expectation of the $K$-function under the assumption of complete spatial randomness and non-spatially dependent marks (in red). Finally, the fourth column reports $\kappa_{\tf}(r)$. In the first and third rows, the marks are spatially dependent.

The subsequent hypotheses come as follows.

\begin{enumerate}
\item[]
\begin{equation*}
\begin{aligned}
\mathrm{H}2 \; &= 
\begin{cases}
\mathcal{H}_{0}: & \text{spatial homogeneity of the point pattern} \\
\mathcal{H}_{1}: & \text{spatial inhomogeneity of the point pattern}
\end{cases}
\\[1em]
\mathrm{H}3 \; &= 
\begin{cases}
\mathcal{H}_{0}: & \text{spatially independent marks }  \\
\mathcal{H}_{1}: & \text{spatially dependent marks }
\end{cases}
\end{aligned}
\end{equation*}
\end{enumerate}

This framework decomposes the overall assessment into targeted hypotheses, providing a detailed characterization of spatial and mark-dependent structures. In the following subsection, we extend this methodology to a \textit{local} setting, defining hypotheses and test statistics for specific spatial regions to detect localized departures from homogeneity and independence.

To test any of the above hypotheses, suitable test statistics are needed.
The proposed test statistics are based on a discrepancy measure of two $K$-functions.
As is well established and can be seen in Figure \ref{fig:enter-label}, particular behaviors of the $K$-functions and of the mark correlation function correspond to specific structural properties of the point patterns. Therefore, it is intuitive to build a test statistic based on \eqref{eq:InhKmmEst}.

\begin{figure}[!htb]
    \centering
    \includegraphics[width=\textwidth]{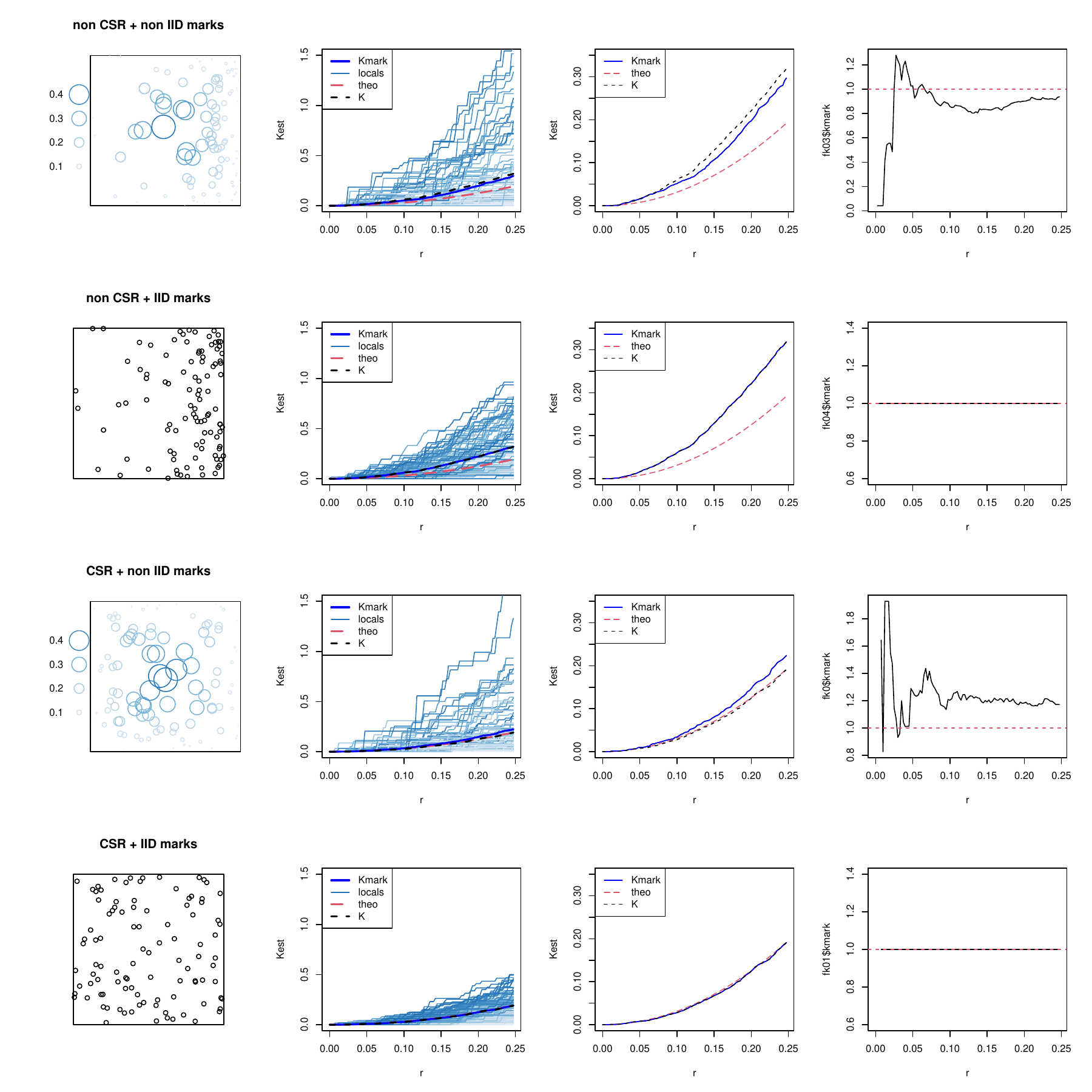}
    \caption{First column: spatial displacement of the point patterns, either inhomogeneous (first and second rows), or homogeneous (third and fourth rows). Second column: Local mark-weighted $K$-functions (in light blue), global mark-weighted $K$-function (in blue), $K$-function (in black), and expectation of the $K$-function under the assumption of complete spatial randomness and non-spatially dependent marks (in red).  Third column: Mark-weighted $K$-functions (in blue), $K$-function (in black), and expectation of the $K$-function under the assumption of complete spatial randomness and non-spatially dependent marks (in red). Fourth column: $\kappa_{\tf}(r)$. First and third rows: the marks are spatially dependent.}
    \label{fig:enter-label}
\end{figure}

Our proposal is therefore of the form
\begin{equation}
 \mathcal{T}=\int_A 
\frac{\big(\hat{K}_{\tf}^{\mathrm{inhom}}(r)-\mathds{E}[\hat{K}_{\tf}^{\mathrm{inhom}}(r)]\big)^2}{\mathds{E}[\hat{K}_{\tf}^{\mathrm{inhom}}(r)]}  \de r,  
\label{eq:chi2}
\end{equation}
where the expectation $\mathds{E}[\hat{K}_{\tf}^{\mathrm{inhom}}(r)]$ depends on the specific hypothesis being tested.

To perform all the above, it remains to determine what the expectation of the marked $K$-function estimator \eqref{eq:InhKmmEst} corresponds to under any of the three alternative hypotheses.

It holds that 
$$K_{\tf}(r) = \mathds{E}\Big[\frac{1}{\lambda(x) c_{\tf}}
 \tf \left(
    m(x),m(y)
    \right)
 \1\lbrace d(x,y)\leq r \rbrace
|
    x \in X
\Big]=$$ $$=\frac{1}{ c_{\tf}}\mathds{E}\Big[\frac{1}{\lambda(x)}
 \tf \left(
    m(x),m(y)
    \right)\1\lbrace d(x,y)\leq r \rbrace
\Big] =$$
$$ =\frac{1}{c_{\tf}}\int_{\mathds{R}^2}  \int_{\mathds{R}}\int_{\mathds{R}} \frac{\tf (
    m(x),m(y)
    )\1\lbrace d(x,y)\leq r \rbrace}{\lambda(x)\lambda(y)}  \lambda(x)\lambda(y)\de x\de y \nu(\de m(x))\nu(\de m(y))=$$
    $$ =\frac{1}{c_{\tf}}\int_{\mathds{R}^2}  \int_{\mathds{R}}\int_{\mathds{R}} \tf (
    m(x),m(y)
    )\1\lbrace d(x,y)\leq r \rbrace \de x\de y \nu(\de m(x))\nu(\de m(y))=$$
    $$= \frac{1}{c_{\tf}}\int_{\mathds{R}^2}   \1\lbrace d(x,y)\leq r \rbrace \de x\de y \int_{\mathds{R}}\int_{\mathds{R}}\tf (
    m(x),m(y)
    )\nu(\de m(x))\nu(\de m(y))$$

Knowing that, 
under the homogeneity assumption, $$\mathds{E}\Biggl[\int_{\mathds{R}^2}   \1\lbrace d(x,y)\leq r \rbrace \de x\de y \Biggr] =  \pi r^2,$$ and that, under mark independence, 
\[
\int_{\mathds{R}} \int_{\mathds{R}} \tf \left( m(x), m(y) \right) \nu(\de m(x)) \nu(\de m(y)) =  c_{\tf},
\]
the following holds.

For homogeneous patterns with non-spatially dependent marks,  $\ee[\hat{K}_{\tf}^{\mathrm{inhom}}(r)] = \pi r^2$, while,
 for inhomogeneous patterns with independent marks, 
$\mathds{E}[\hat{K}_{\tf}^{\mathrm{inhom}}(r)] = K(r),$
and for homogeneous marked patterns,  $\mathds{E}[\hat{K}_{\tf}^{\mathrm{inhom}}(r)] = \pi r^2\frac{c_{\tf}(r)}{c_{\tf}} = \pi r^2 \kappa_{\tf}(r).$

Imputing a constant intensity, we get $\ee[\hat{K}_{\tf}(r)] = \pi r^2$, 
$\mathds{E}[\hat{K}_{\tf}(r)] = K(r),$
and $\mathds{E}[\hat{K}_{\tf}(r)] = \pi r^2\frac{c_{\tf}(r)}{c_{\tf}} = \pi r^2 \kappa_{\tf}(r).$

This gives the three test statistics
$$
 \mathcal{T}_1=\int_{ A}
\frac{\big( \widehat{K}_{\tf}(r) - \pi r^{2} \big)^{2}}
{\pi r^{2}}\mathrm{d}r,
\quad 
 \mathcal{T}_2=\int_{A}
\frac{\big( \widehat{K}_{\tf}(r)
          - K(r) \big)^{2}}
{ K(r)}\mathrm{d}r, 
\quad 
 \mathcal{T}_3=\int_{A}
\frac{\big( \widehat{K}_{\tf}(r)
          - \pi r^{2} \kappa_{\tf}(r) \big)^{2}}
{\pi r^{2} \kappa_{\tf}(r)}\mathrm{d}r.
$$

In particular, $\mathcal{T}_1$ is employed for testing $\mathrm{H}1$, $ \mathcal{T}_2$ is used to  mark dependence ($\mathrm{H}2$), and  $\mathcal{T}_3$ to test test point dependence ($\mathrm{H}3$).

We resort to Monte Carlo simulations to obtain the distribution of the test statistics under any of the null hypotheses we aim to test and to obtain the critical value, consequently.
In particular, we generate 100 new patterns and consider the quantile of the test statistic $\mathcal{T}$ as the threshold value, rejecting whenever  $\mathcal{T}$ is higher than the quantile of the simulated distribution.

\subsection{Local test} \label{sec:local_test}

We know that the conclusion drawn from the application of the above-mentioned global test pertains to the whole analysed process, indicating whether all the marks are randomly labelled or not.
Motivated by the will to further detect the specific points, and regions, where the marks do depend on the other marked points, 
we propose some \textit{local tests}. As in \cite{d2024local}, the main idea is to run a global test on each point of the analysed pattern using the previously proposed \textit{mark-weighted $K$-functions}, to draw different conclusions about the individual points, based on the obtained $p$-values. 

The systems of hypothesis in the local case are therefore referred to the specific point $x_i$ belonging to the point pattern $\textbf{x}$.

Assuming a null hypothesis of absence of structure in the point and in its mark, we wish to test the following hypotheses:
\begin{enumerate}
\item[]
\begin{equation*}
\begin{aligned}
\mathrm{H}1_{L} \; &= 
\begin{cases}
\mathcal{H}_{0}: & \text{local homogeneity with independent marks around } x_i \\
\mathcal{H}_{1}: & \text{deviation from local homogeneity and/or mark independence around } x_i
\end{cases} 
\\[1em]
\mathrm{H}2_{L} \; &= 
\begin{cases}
\mathcal{H}_{0}: & \text{local homogeneity of the point pattern around } x_i \\
\mathcal{H}_{1}: & \text{local inhomogeneity of the point pattern  around } x_i
\end{cases}
\\[1em]
\mathrm{H}3_{L} \; &= 
\begin{cases}
\mathcal{H}_{0}: & \text{local independent marks around } x_i \\
\mathcal{H}_{1}: & \text{local dependent marks around } x_i
\end{cases}
\end{aligned}
\end{equation*}
\end{enumerate}

The local hypothesis $\mathrm{H}1_{L}$ is used to detect the presence of spatial structure at a local scale. 
If $\mathcal{H}_{0}$ is rejected for a given point, additional local tests $\mathrm{H}2_{L}$ and $\mathrm{H}3_{L}$ are performed to determine whether the detected structure is associated with the marks, the spatial distribution of the points, or both.

Also the local test statistics is based on a discrepancy measure of two (local) $K$-functions, and it informs us about which points contributed to the possible global statement about the rejection of the null hypothesis.

In the most general form, this can be tested through the test statistics
\begin{equation}
 \mathcal{T}^{(i)}=\int_A 
\frac{\big(\hat{K}_{\tf,i}^{\mathrm{inhom}}(r)-\mathds{E}[\hat{K}_{\tf,i}^{\mathrm{inhom}}(r)]\big)^2}{\mathds{E}[\hat{K}_{\tf,i}^{\mathrm{inhom}}(r)]}  \de r,  
\label{eq:chi2_local}
\end{equation}
that is, providing test statistics for each $i$-th point of the pattern. In this way, it is possible to obtain a p-value for each point, basically allowing us to assess which points contributed to the rejection of the \lq \lq global" hypothesis, and in particular, allowing us to assess the presence of specific spatial regions where the hypothesis holds or not.

\begin{prop}
Under each of the three reference hypotheses described above, the expected value of the local marked inhomogeneous $K$-function equals that of the global one, that is,
\[
\mathds{E}[\hat{K}_{\tf,i}^{\mathrm{inhom}}(r)] = \mathds{E}[\hat{K}_{\tf}^{\mathrm{inhom}}(r)].
\]
   
\end{prop}

\begin{proof}
Taking the conditional expectation only with respect to $j$, and assuming $\tf_i (
    m(x_i),m(y)
    ) = m(x_i)m(y)$, it holds that 
$$K_{\tf,i}^{\mathrm{inhom}}(r)=
 \frac{1}{c_{\tf,i}\lambda(x_i)}\int_{\mathds{R}^2}  \int_{\mathds{R}} \frac{\tf_i (
    m(x_i),m(y)
    )\1\lbrace d(x_i,y)\leq r \rbrace}{\lambda(y)}  \lambda(y)\de y \nu(\de m(y))=$$
    $$ =\frac{1}{c_{\tf,i}\lambda(x_i)}\int_{\mathds{R}^2}  \int_{\mathds{R}}\tf_i (
    m(x_i),m(y)
    )\1\lbrace d(x_i,y)\leq r \rbrace \de y \nu(\de m(y))=$$
    $$= \frac{1}{c_{\tf,i}\lambda(x_i)}\int_{\mathds{R}^2}   \1\lbrace d(x_i,y)\leq r \rbrace \de y \ m(x_i)\int_{\mathds{R}}m(y)
    \nu(\de m(y))$$

Knowing that, 
under the homogeneity assumption, $$\mathds{E}\Biggl[\int_{\mathds{R}^2}   \1\lbrace d(x_i,y)\leq r \rbrace \de y \Biggr] = \lambda(x_i) \pi r^2,$$ and that, under mark independence, 
\[
\mathds{E}\Biggl[ m(x_i)\int_{\mathds{R}}  m(y)\nu(\de m(y))\Biggr] = \mathds{E}\Biggl[ m(x_i)\mathds{E}[m(y)] \Biggr] = \mathds{E}[ m(x_i)\mu_y ] = \mu^2 = c_{\tf,i},
\]
the following holds.
     
For homogeneous patterns with non-spatially dependent marks ,  $$\ee[\hat{K}_{\tf,i}^{\mathrm{inhom}}(r)] = \ee[\hat{K}_{\tf}^{\mathrm{inhom}}(r)] = \pi r^2.$$ 
For inhomogeneous patterns with independent marks, 
$$\mathds{E}[\hat{K}_{\tf,i}^{\mathrm{inhom}}(r)] = \mathds{E}[\hat{K}_{\tf}^{\mathrm{inhom}}(r)] = K(r).$$
 For homogeneous marked patterns,  $$\mathds{E}[\hat{K}_{\tf,i}^{\mathrm{inhom}}(r)] = \mathds{E}[\hat{K}_{\tf}^{\mathrm{inhom}}(r)] =\pi r^2 \kappa_{\tf}(r).$$
\end{proof}

It consequently holds that 
$$
\mathcal{T}^{(i)}_1=\int_A 
\frac{\big(\hat{K}_{\tf,i}(r)-\pi r^2\big)^2}{\pi r^2} \de r ,
\quad
\mathcal{T}^{(i)}_2=\int_A  
\frac{\big(\hat{K}_{\tf,i}(r)-K(r)\big)^2}{K(r)}\de r,
\quad
\mathcal{T}^{(i)}_3=\int_A  
\frac{\big(\hat{K}_{\tf,i}(r)-\pi r^2 \kappa_{\tf}(r)\big)^2}{\pi r^2 \kappa_{\tf}(r)}\de r,
$$
for the three hypotheses, respectively.

Also for these local scenarios, $\mathcal{T}^{(i)}_1$ is employed for testing $\mathrm{H}1_{L}$, $\mathcal{T}^{(i)}_2$ for $\mathrm{H}2_{L}$, and $\mathcal{T}^{(i)}_3$ for $\mathrm{H}3_{L}$.

\section{Simulation study}\label{sec:sims}

This section is devoted to the assessment of the proposed test statistics. In particular, Section \ref{sec:glob} explores the global scenarios assessing the test performance in terms of power of the test, while Section \ref{sec:loc} explores the local scenarios by assessing the ability to recover \lq \lq local points" through classification rates. 

\subsection{Global test}\label{sec:glob}

For each of the three hypotheses, we simulate 100 spatial point patterns in the unit square under the alternative hypotheses and compute the power of the test. 

First, we consider a homogeneous Poisson process with constant intensity $\lambda \in \{25, 50, 100\}$, directly representing the expected number of points in the unit square window $W=[0,1]^2$. 

Then, we consider an inhomogeneous Poisson process with a linear intensity function of the form
$$
\lambda(x,y)= 10 + \alpha x, \qquad
(x,y) \in [0,1]^2,
$$ 
where we let $\alpha \in \{30, 80, 180
\}$, leading to having $\mathbb{E}[N]= \{25, 50, 100
\}$ as expected numbers of points.

To simulate spatially dependent mark structures, we assign the mark to the point $x_i$ of the point pattern $\textbf{x}$ as the shortest distance $d(x_i,W)$ from the point to the boundary of the window $W$.

Table~\ref{tab:power} presents the estimated power of the global test across a range of simulated scenarios, averaged over 100 replicates. Each setting combines different spatial point pattern structures (either homogeneous or inhomogeneous Poisson processes), varying the degree of mark structure (based on the distance from the boundary), and sample sizes reflecting the expected number of points $\mathbb{E}[N]$. To make the mark dependence vary, we further consider increasing values of the power $h$ of the distance from the boundary $d(x_i,W)^h$. In particular, we will explore three degrees of marks dependence ($h = 1, 2, 3$).

\begin{table}[!htb]
\centering
\caption{Power of the global test over the three hypotheses, number of points, and power of the mark, averaged over 100 replicates.}
\label{tab:power}

\begin{tabular}{cr|ccc}
\toprule
&&\multicolumn{3}{c}{\textbf{Hypothesis}}  \\
\textbf{$\mathbb{E}[N]$} & \textbf{h} & $\mathrm{H}1$ & $\mathrm{H}2$ & $\mathrm{H}3$ \\
\midrule
\multirow{3}{*}{25}  & 1 & 0.40 & 0.63 & 0.09 \\
   & 2 & 0.73 & 0.75 & 0.51 \\
   & 3 & 0.84 & 0.86 & 0.95 \\
 \midrule
 \multirow{3}{*}{50}  & 1 & 0.90 & 0.91 & 0.11 \\
   & 2 & 0.99 & 0.99 & 0.53 \\
   & 3 & 1.00 & 1.00 & 0.94 \\
 \midrule
 \multirow{3}{*}{100} & 1 & 1.00 & 0.99 & 0.10 \\
  & 2 & 1.00 & 1.00 & 0.54 \\
  & 3 & 1.00 & 1.00 & 0.94 \\
\bottomrule
\end{tabular}
\end{table}

Table \ref{tab:power} reports the empirical power of the proposed global test across the three hypothesis systems, three different point structures ($h = 1, 2, 3$), and varying expected number of points ($\mathbb{E}[N]$). Each entry represents the proportion of 100 Monte Carlo replicates in which the null hypothesis was rejected. 
Overall, the power increases with the expected number of points, indicating that larger point patterns improve the ability of the test to detect departures from the null hypothesis. In particular, $\mathrm{H}1$ and $\mathrm{H}2$ display a rapid increase in power as $\mathbb{E}[N]$ grows, reaching values close to one already for moderate sample sizes. In contrast, $\mathrm{H}3$ shows a lower sensitivity for small patterns, especially when $h=1$, although its power also improves as the number of points increases. The lower power of $\mathrm{H}3$, if compared to that of $\mathrm{H}1$ and $\mathrm{H}2$, is likely due to the difficulty of the $K$-function in correctly identifying mark structure over point structure.

\subsection{Local test}\label{sec:loc} 

To assess the ability of local tests in detecting spatially localised patterns in the marks, we simulate spatial point patterns in the unit square $W = [0,1]^2$, where only a small fraction of points exhibit distinctively different marks or clustering behaviour depending on the scenario. The aim is to create scenarios where the global distribution of points and marks remains approximately independent, while local deviations can be detected.
We consider three different levels of global process intensities, corresponding to $\mathbb{E}[N] \in \{25, 50, 100\}$ expected points, generated from a homogeneous Poisson process in $W$. 

When simulating local point structure (i.e. $\mathcal{H}_0$ of hypothesis $\mathrm{H}2_{L}$), we generate a configuration obtained by superimposing a homogeneous background process with a clustered component. Specifically, a background pattern is first simulated from a homogeneous Poisson process with intensity $\lambda$. Then, a clustered point pattern is generated using a Thomas process with parent intensity $\kappa$, mean number of offspring per parent $\mu$, and Gaussian displacement with standard deviation $\sigma = 0.03$. The two patterns are superimposed to obtain the final configuration $X$, resulting in a pattern that contains both randomly scattered points and small spatial clusters. The superposition allows us to identify which points belong to the clustered component. The parameters $\lambda$, $\kappa$, and $\mu$, are selected in order to give $30\%$ of points belonging to the clustered component. 
After generating the locations, marks are assigned independently of the point locations. In particular, all points receive i.i.d.\ marks drawn from a uniform distribution on $[0,1]$. The indicator of cluster membership is stored separately and used only for evaluation purposes. This allows us to assess whether the proposed local test is able to detect points belonging to spatial clusters even when the marks themselves do not carry additional information about the clustering structure.

When simulating scenarios with local mark dependence (i.e. $\mathcal{H}_0$ of hypothesis $\mathrm{H}3_{L}$), we randomly select a small number $k$ of points $\{y_1, \dots, y_k\}$ from the generated configuration, and define local neighborhoods around them as balls $B(y_j, r)$ of fixed radius $r = 0.05$. All points falling in these neighborhoods form the set $Z$, representing the regions with locally modified marks. The remaining points, denoted by $X \setminus Z$, constitute the background pattern.
Marks are assigned as follows. The background points in $X \setminus Z$ receive i.i.d. marks from a standard uniform distribution. Each selected center $y_j$ is assigned a mark $m(y_j) \sim \mathcal{N}(\mu, \sigma^2)$, with $\mu = 5$ and $\sigma^2 = 1$. This ensures that only a few small areas in the window exhibit higher-than-usual marks, while the rest of the pattern remains indistinguishable from a homogeneous marked process. Therefore, a point is considered a true positive if it belongs to $Z$ and is correctly identified by the local test as significant, while false positives correspond to points from $X \setminus Z$ that are incorrectly classified as significant. 

Finally, for the scenario combining both local point inhomogeneity and local mark dependence (i.e. $\mathcal{H}_1$ of hypotheses $\mathrm{H}1_{L}$, $\mathrm{H}2_{L}$, and $\mathrm{H}3_{L}$), the spatial configuration is generated as in the previous setting with clustered locations by superimposing a homogeneous Poisson background process with intensity $\lambda = 35$ and a Thomas cluster and marks are then assigned conditionally on the cluster membership. Background points receive i.i.d.\ marks from a $\mathrm{Unif}(0,1)$ distribution, while points belonging to the clustered component receive marks generated from a Gaussian. In this way, clustered regions correspond simultaneously to areas with higher point concentration and locally increased marks, thus combining the two types of departures from the null hypothesis considered in the previous scenarios. The three scenarios are represented in Figure \ref{fig:local_show}.

\begin{figure}[!htb]
	\centering
	\subfloat[Local point structure]
    {\includegraphics[width=.33\textwidth]{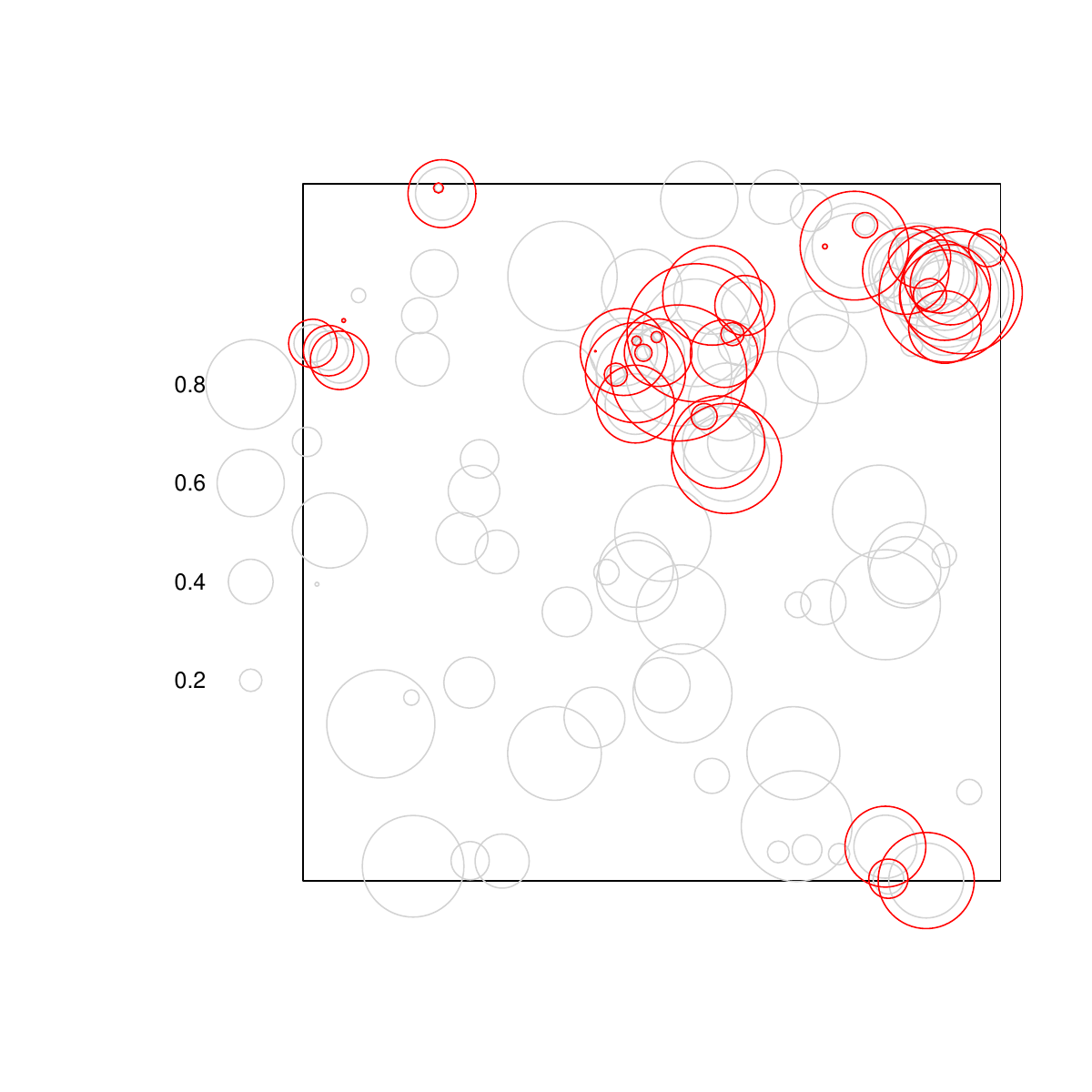}}  
	\subfloat[Local mark structure]
    {\includegraphics[width=.33\textwidth]{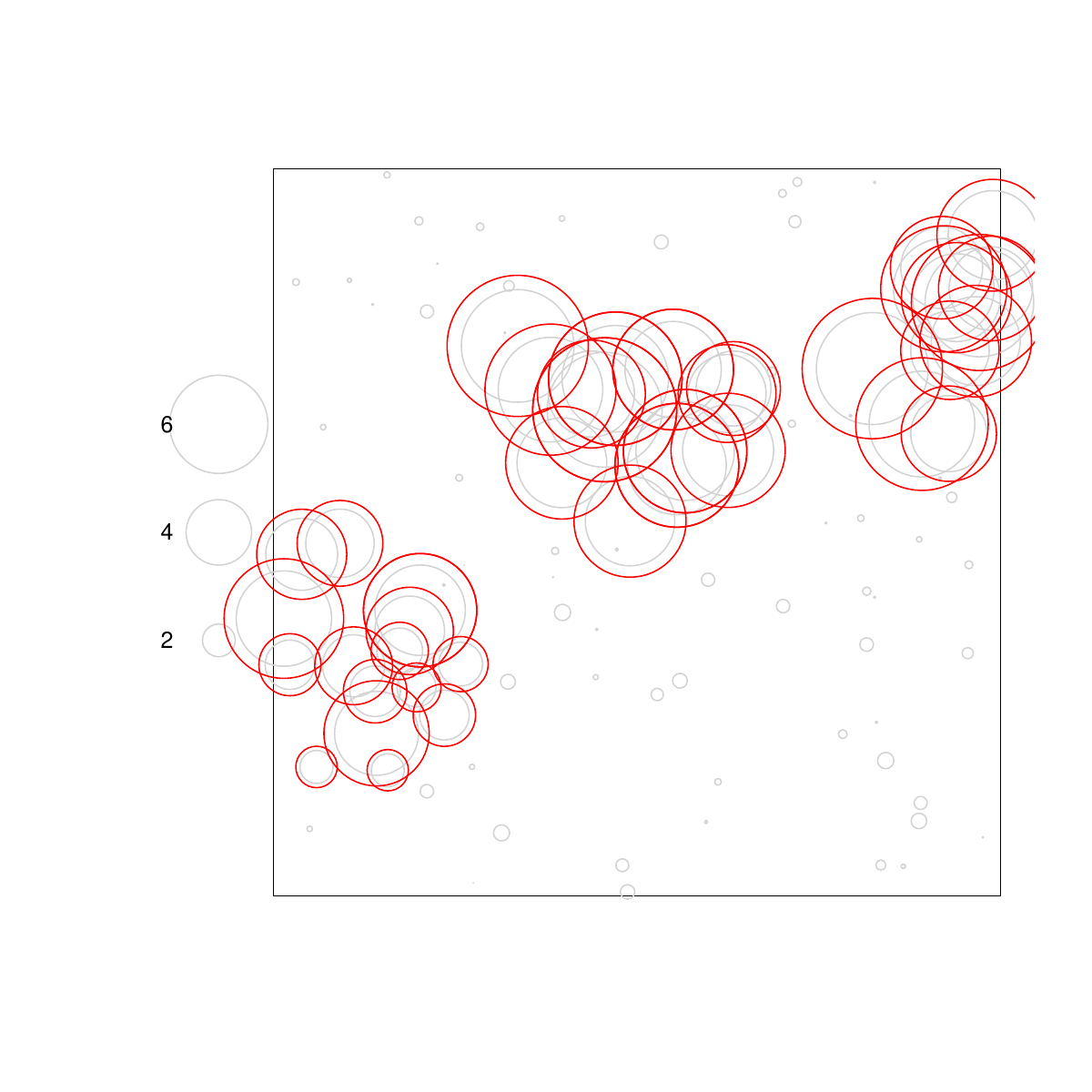}}
    \subfloat[Local point and mark structure]
    {\includegraphics[width=.33\textwidth]{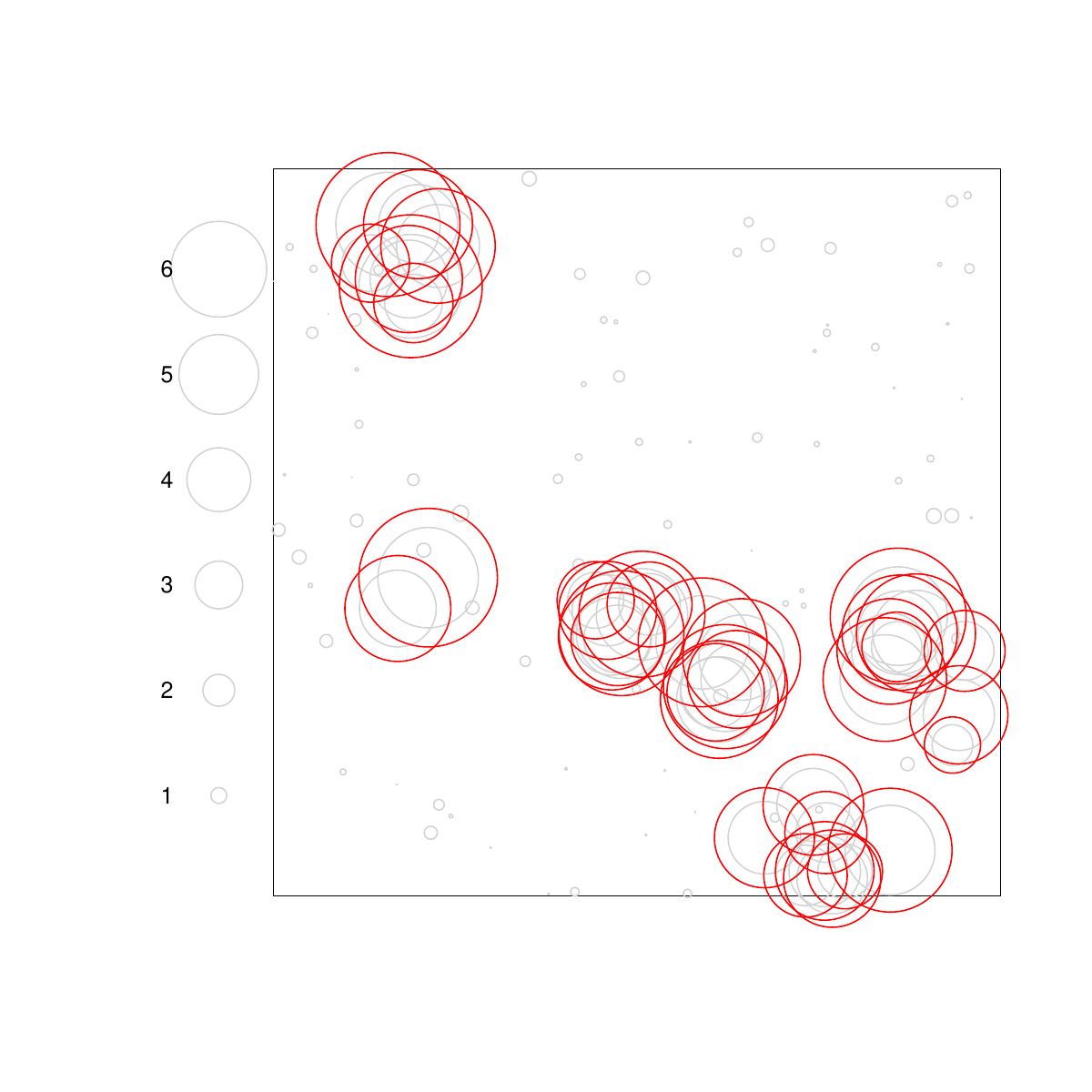}}
	\caption{Simulated patterns for the local scenarios.}
	\label{fig:local_show}
\end{figure}

We repeat the simulation 100 times for each expected number of points ($\mathbb{E}[N]$), and compute standard classification metrics such as true positive rate (TPR), false positive rate (FPR), and accuracy (ACC). 
The results for the local tests, shown in Table \ref{tab:power_local}, provide information on the classification performance. 

\begin{table}[!htb]
\centering
\caption{Classification rates of the local test over the three hypotheses, number of points, and power of the mark, averaged over 100 replicates.}
\label{tab:power_local}

\begin{tabular}{cc|ccc}
\toprule
 & & \multicolumn{3}{c}{\textbf{Hypothesis}} \\
\textbf{$\mathbb{E}[N]$} & \textbf{Metric} 
& $\mathrm{H}1_{L}$
& $\mathrm{H}2_{L}$ 
& $\mathrm{H}3_{L}$ \\
\midrule

\multirow{3}{*}{25}
 & TPR & 0.55 & 0.45 & 0.29 \\
 & FPR & 0.00 & 0.00 & 0.00 \\
 & ACC & 0.81 & 0.76 & 0.70 \\

\midrule
\multirow{3}{*}{50}
 & TPR & 0.68 & 0.60 & 0.32 \\
 & FPR & 0.00 & 0.00 & 0.00 \\
 & ACC & 0.89 & 0.86 & 0.78 \\

\midrule
\multirow{3}{*}{100}
 & TPR & 0.79 & 0.70 & 0.34 \\
 & FPR & 0.00 & 0.00 & 0.00 \\
 & ACC & 0.93 & 0.90 & 0.79 \\

\bottomrule
\end{tabular}
\end{table}

Each entry is averaged over 100 replicates.
Several observations can be drawn. The true positive rate (TPR) generally increases with the number of points, indicating improved sensitivity of the local test in detecting localized departures from the null hypothesis as the sample size grows. The false positive rate (FPR) remains consistently close to zero across all scenarios. As a consequence, the overall accuracy (ACC) also increases with larger $\mathbb{E}[N]$ for all hypotheses. Among the considered scenarios, $\mathrm{H1}_{L}$ and $\mathrm{H2}_{L}$ exhibit the highest TPR and ACC values, whereas $\mathrm{H3}_{L}$ shows lower sensitivity, particularly for smaller patterns, although its performance improves as the number of points increases. Similarly to the global case, the lower performance of $\mathrm{H}3$, if compared to that of $\mathrm{H}1$ and $\mathrm{H}2$, is likely due to the difficulty of the $K$-function in correctly identifying mark structure over point structure.

\section{Applications}\label{sec:appls}

This section is devoted to the application of the proposed tests on two environmental datasets with spatially dependent marks representing different scenarios. In particular, Section \ref{sec:waka} consists of an application to patterns of trees locations marked by their diameter, available from the R \citep{R} package \texttt{spatstat} \citep{baddeley:rubak:tuner:15}, while Section \ref{sec:seismic} presents an application to a clustered pattern of seismic events marked by their magnitude.

\subsection{Forestry data}\label{sec:waka}

The first dataset is a spatial point pattern of 504 trees recorded at Waka National Park, Gabon \citep{balinga2006vegetation,picard2009multi}.
The trees are marked by the tree diameter at breast height (DBH). The survey region is a 100 by 100 square metre. Coordinates are given in metres, while the DBH is in centimetres. Figure \ref{fig:appl1} shows the marked point pattern and the observed unmarked $K$-function, together with its theoretical value under CSR, $\pi r^2$. The $K$-function indicates spatial ramdomness.

\begin{figure}[!htb]
	\centering
	\subfloat[Point pattern]{\includegraphics[width=.33\textwidth]{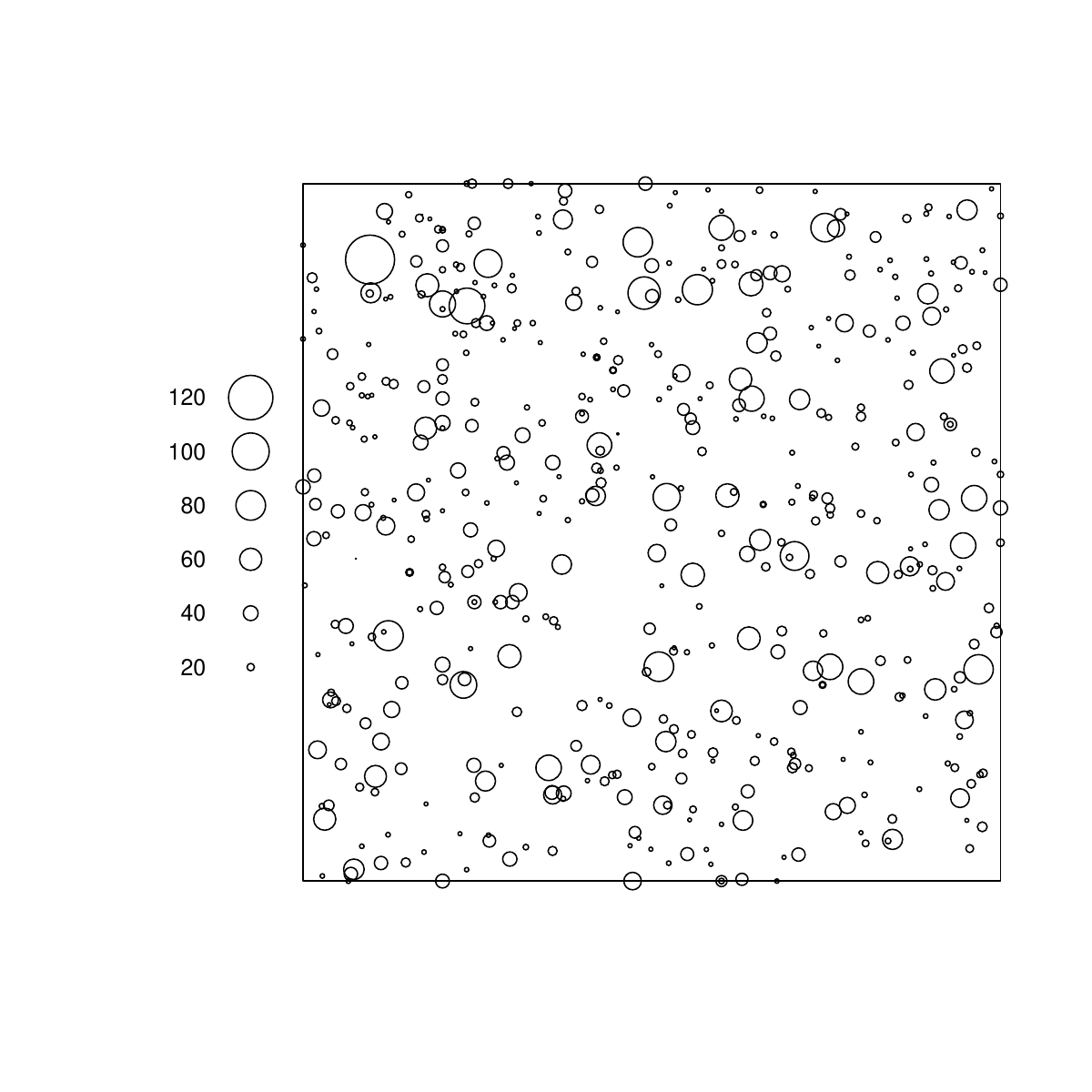}}  
	\subfloat[Unmarked $K$-function]
    {\includegraphics[width=.33\textwidth]{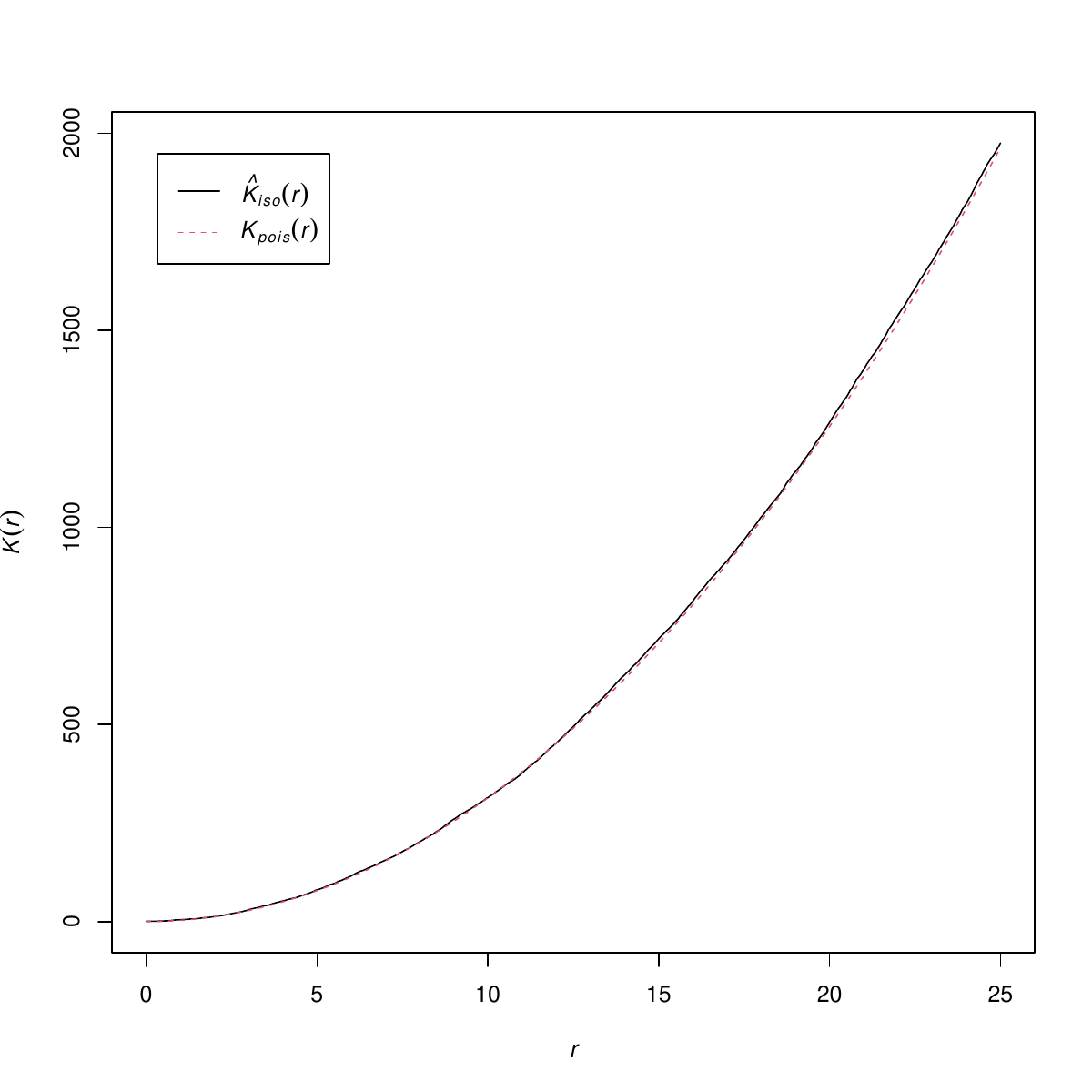}}
	\caption{Point pattern of Waka Trees marked by their diameter (a) and its observed $K$-function (b).}
	\label{fig:appl1}
\end{figure}

We test $\mathrm{H}1$, non rejecting the null hypothesis. We therefore may conclude that the pattern is overall homogeneous and with independent marks. Assessing  $\mathrm{H}1_L$, we get that the significant point are 6\% of the total.

The second dataset records the locations of 126 pine saplings in a Finnish forest marked by their heights and their diameters. Figure \ref{fig:finpines}(a)-(b) shows the marked point pattern and the observed unmarked $K$-function, together with its theoretical value under CSR, $\pi r^2$. The $K$-function indicates some global clustering of points.

\begin{figure}[!htb]
	\centering
	\subfloat[Point pattern]{\includegraphics[width=.33\textwidth]{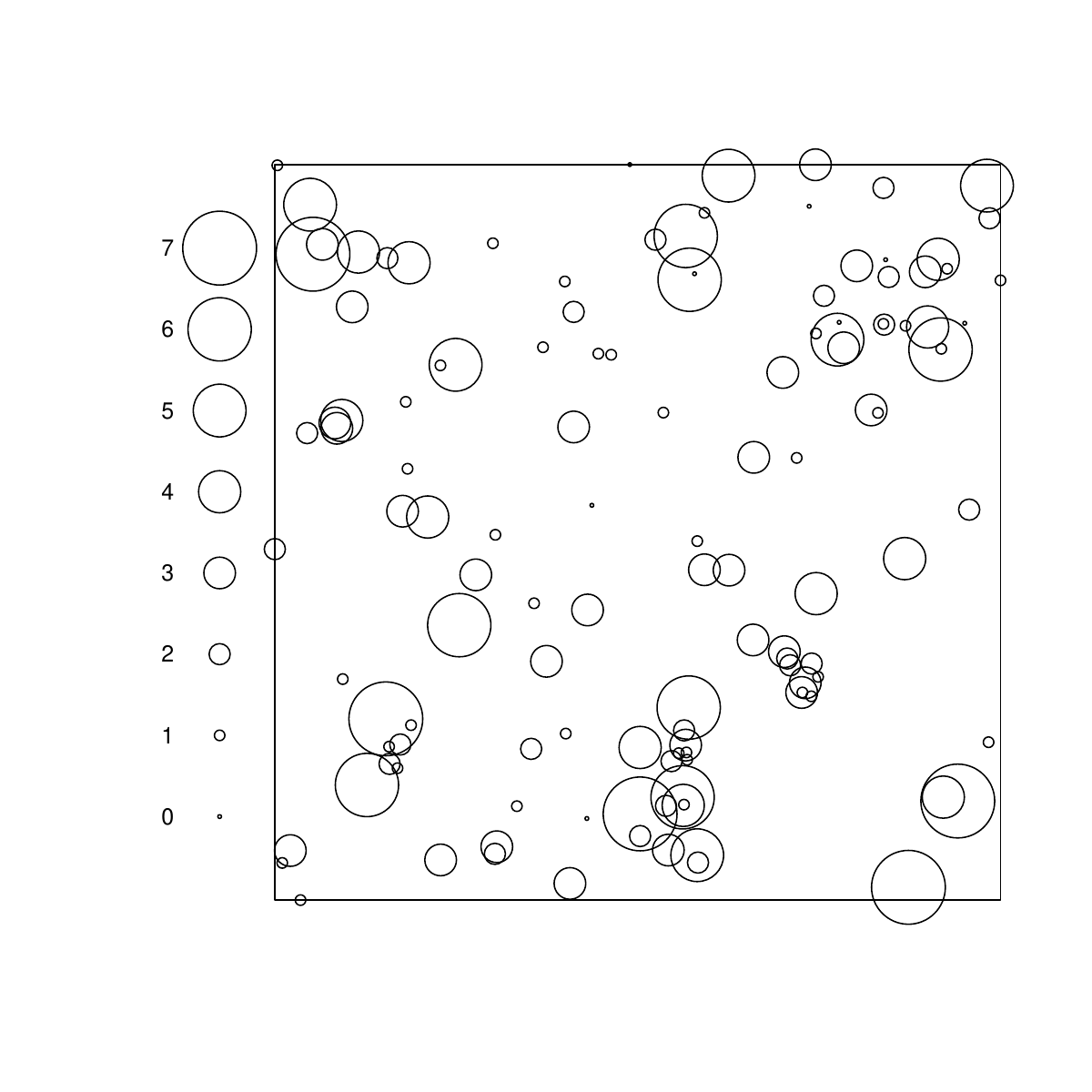}}  
	\subfloat[Unmarked $K$-function]
    {\includegraphics[width=.33\textwidth]{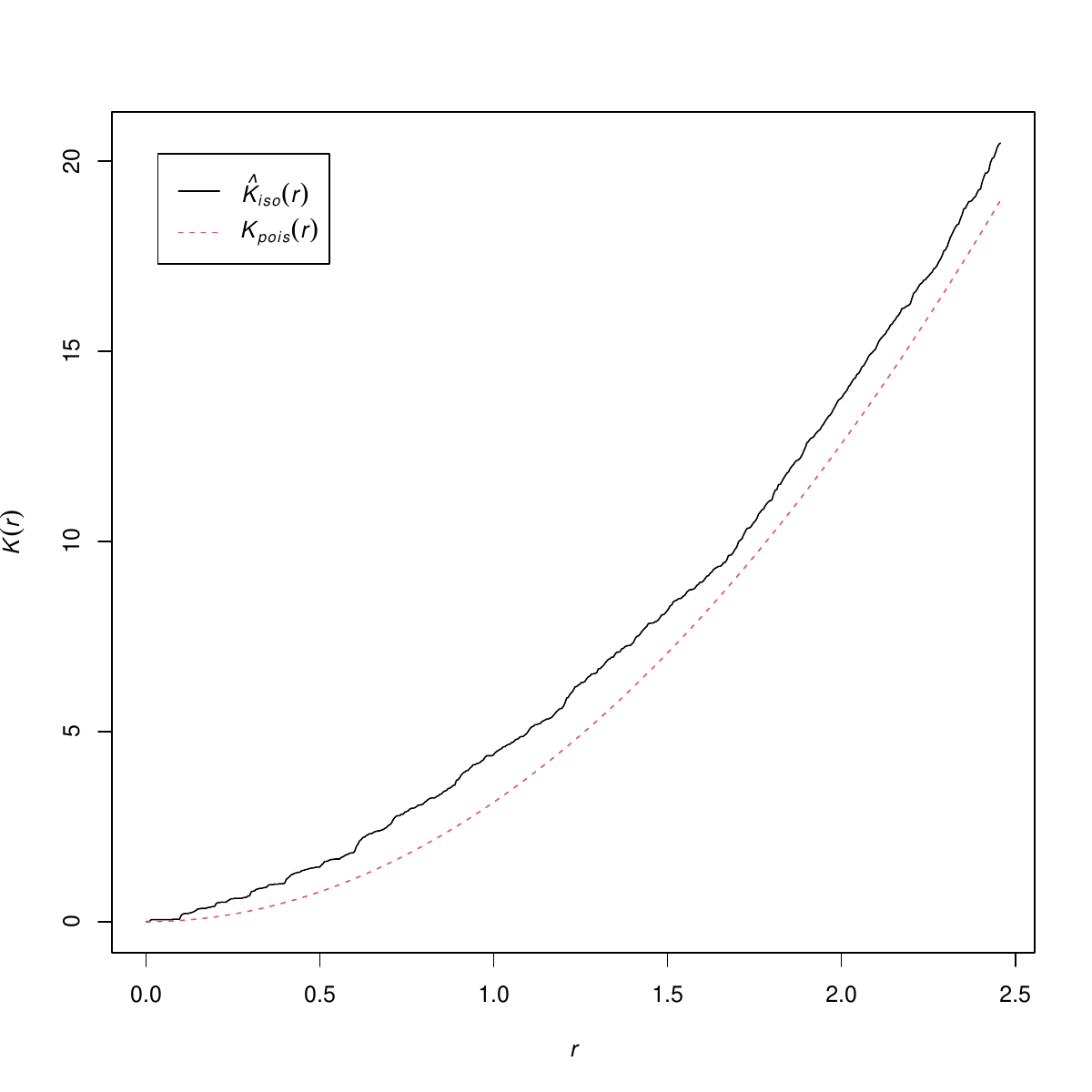}}
    	\subfloat[Significant points]
        {\includegraphics[width=.33\textwidth]{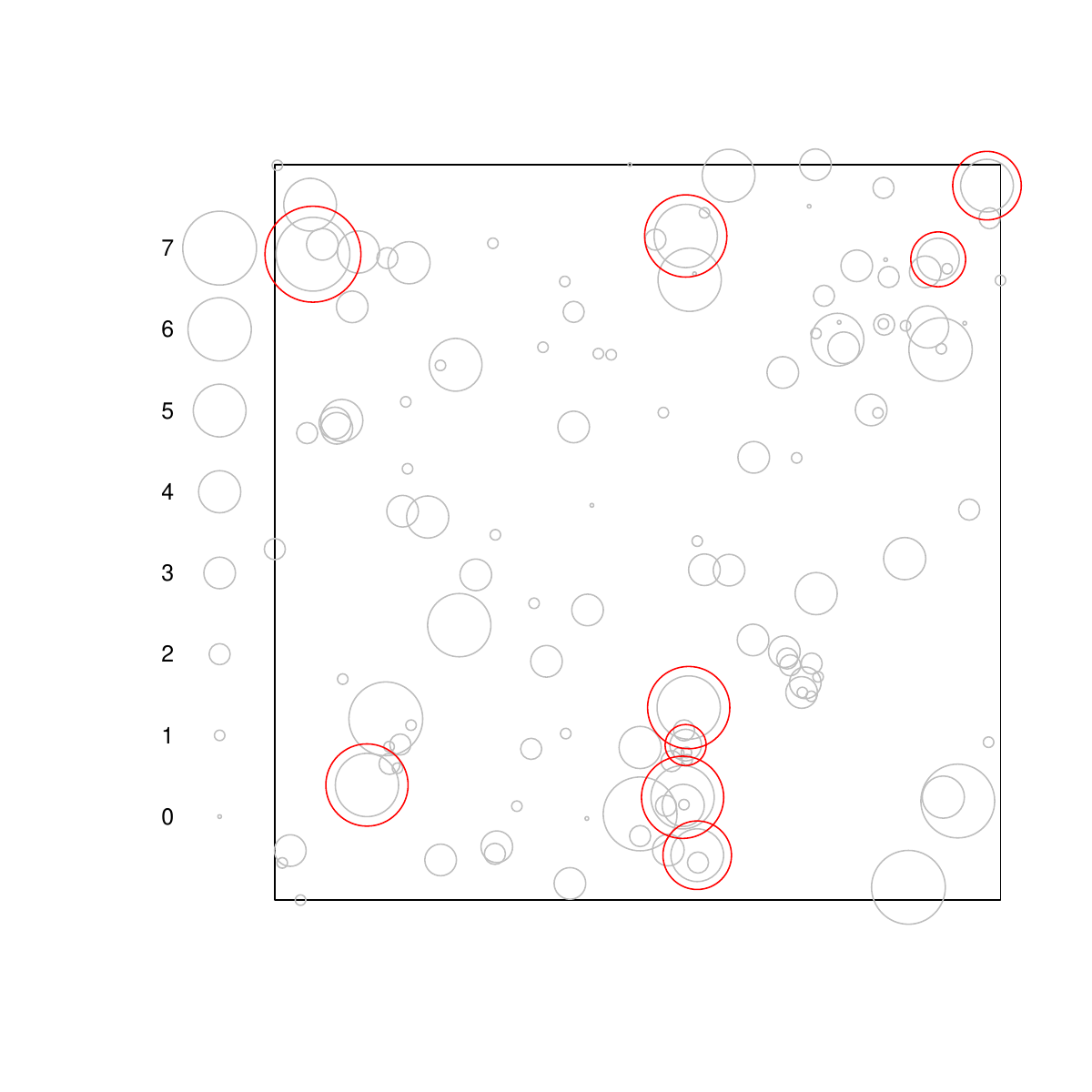}}     
	\caption{Point pattern of pine saplings in Finland marked by their diameter (a) and its observed $K$-function (b). Point pattern of pine saplings in Finland marked by their diameter (c). In red, the significant points of the local test for hypothesis $\mathrm{H}1$.}
	\label{fig:finpines}
\end{figure}

We test $\mathrm{H}1$, rejecting the null hypothesis. We therefore proceed by also testing $\mathrm{H}2$, and not rejecting $\mathcal{H}_1$, we may conclude that the pattern is overall inhomogeneous and with independent marks. Assessing  $\mathrm{H}1_L$, we get that the significant point depicted in Figure \ref{fig:finpines}(c), which likely contributed to the global rejection of the hypothesis. 

\subsection{Earthquake data}\label{sec:seismic}

We analyse a dataset of seismic events named \texttt{italycatalog} available from the  \texttt{etasFLP} (\cite{chiodi2017mixed}) R package. It contains a sample catalog of Italian earthquakes of magnitude at least 3.0 from year 2005 to year 2013. Figure 	\ref{fig:appl4}(a) displays a sample of 2158 earthquakes of the whole dataset of seismic events occurred all over Italy in that time frame.

\begin{figure}[!htb]
	\centering
	\subfloat[{Italy}]{\includegraphics[width=.4\textwidth]{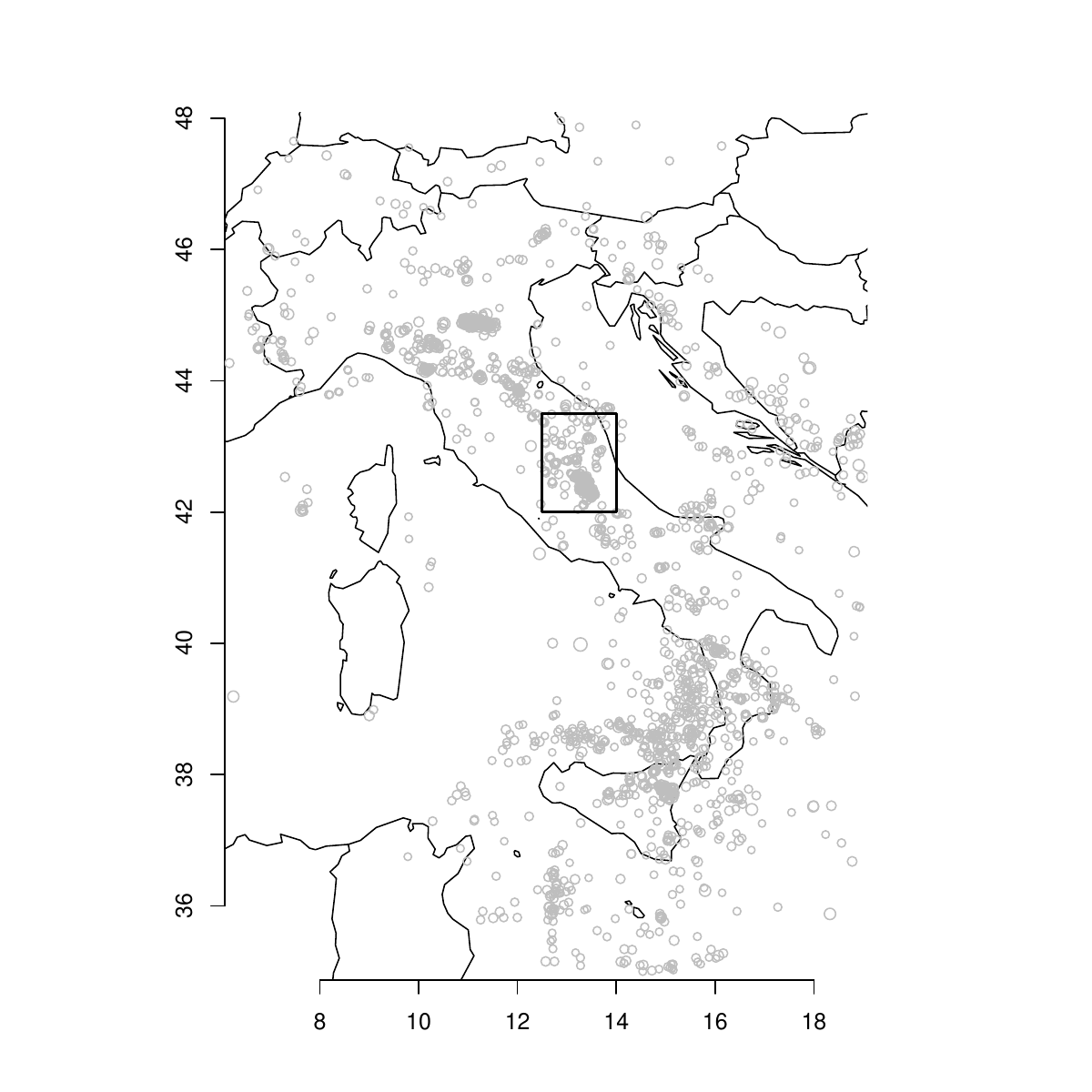}}  
	\subfloat[{Close-up to Abruzzo region}]{\includegraphics[width=.4\textwidth]{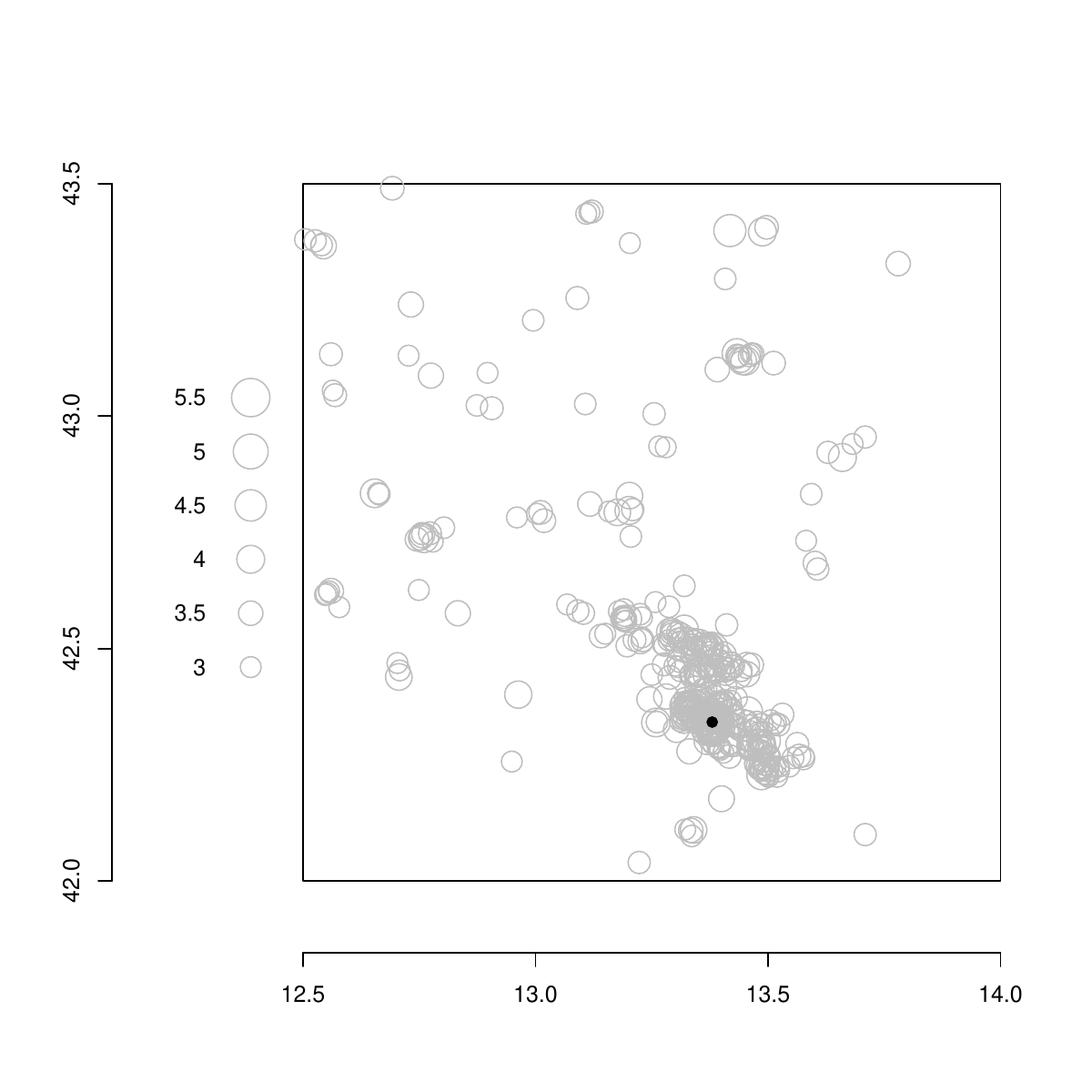}}  \\
    	\subfloat[Time]{\includegraphics[width=.4\textwidth]{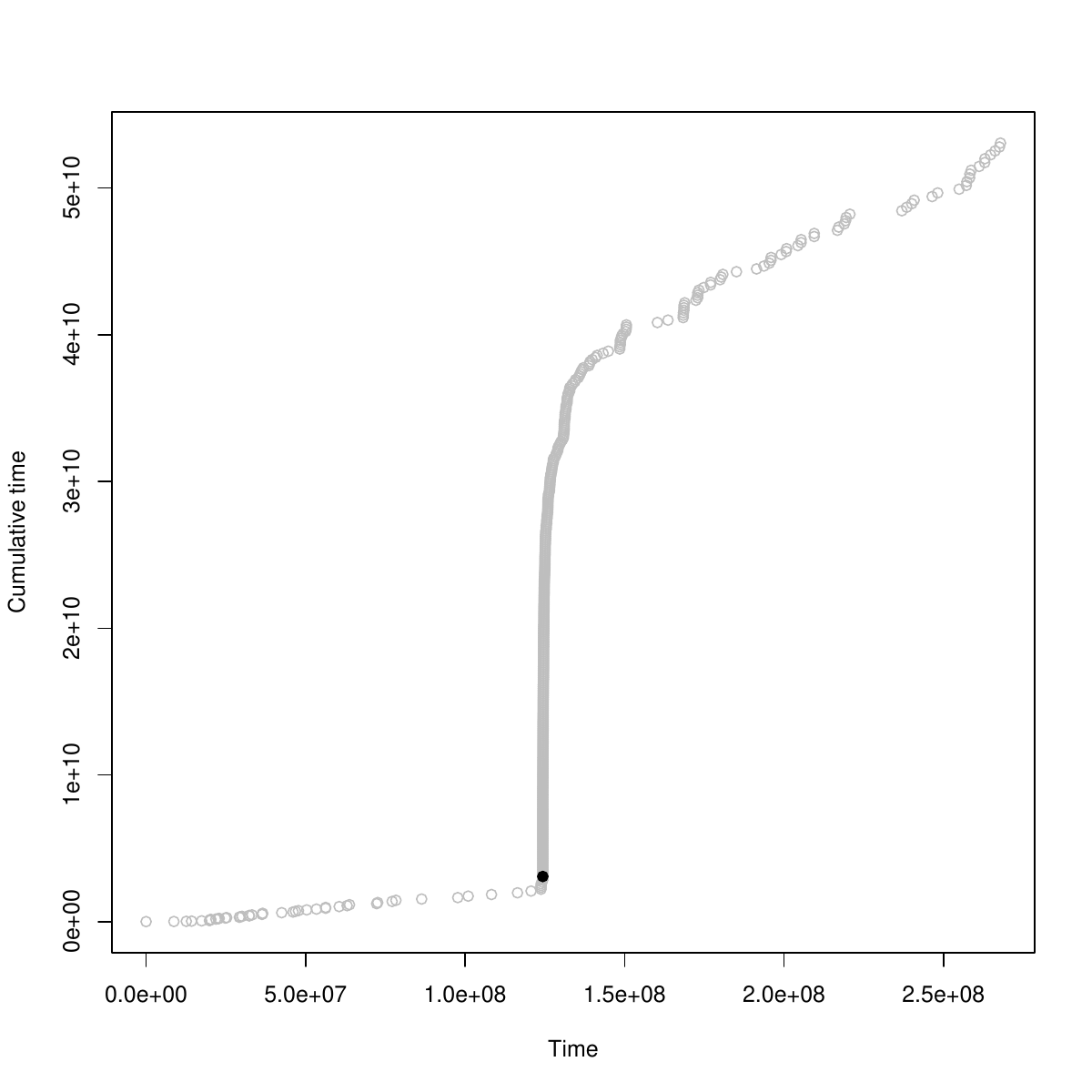}} 
	\subfloat[Unmarked $K$-{function}]{\includegraphics[width=.4\textwidth]{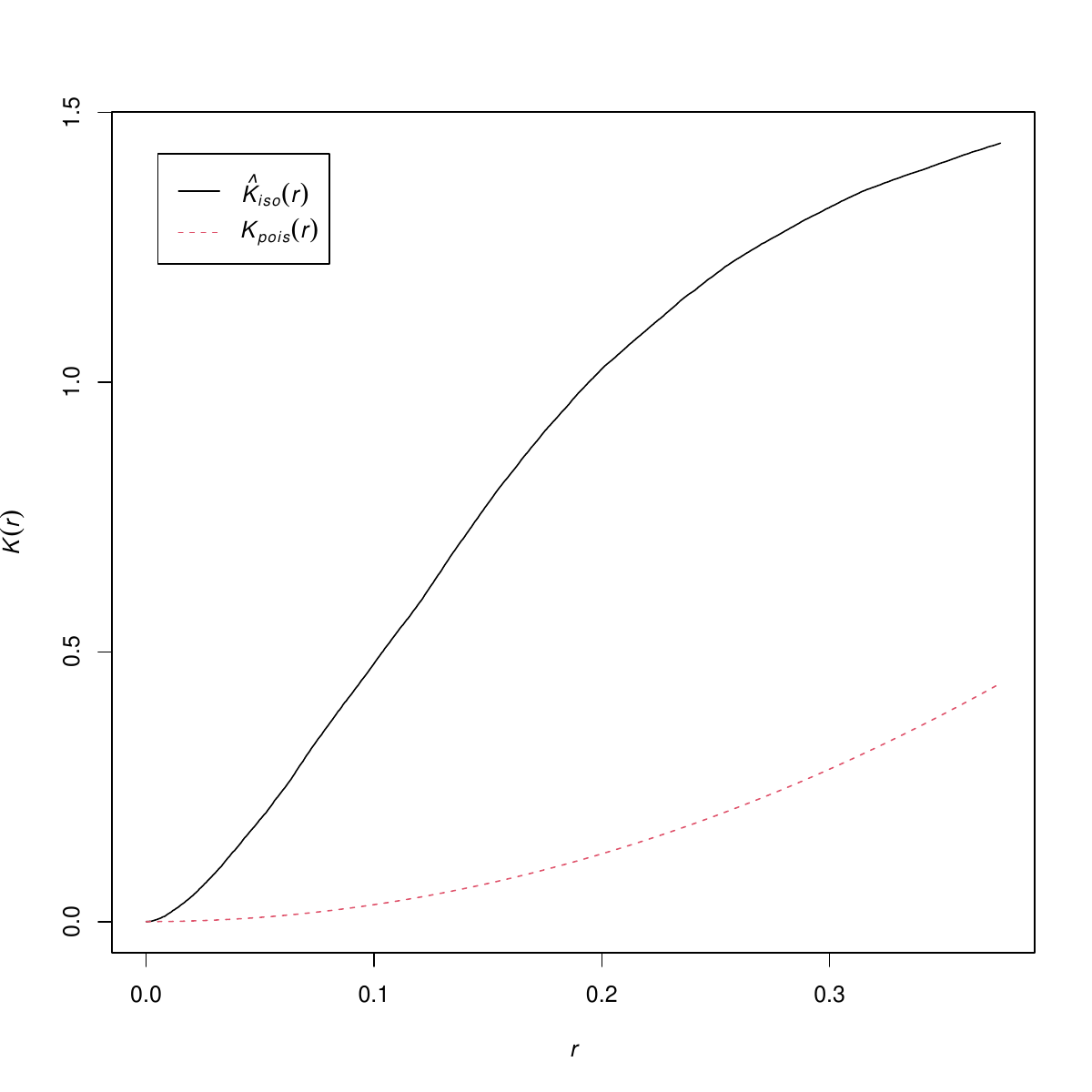}}  
	\caption{Point pattern of earthquake data marked by the magnitude in Italy (a), Abruzzo (b), and its observed $K$-function. The spatial location of the mainshock is in black.}
	\label{fig:appl4}
\end{figure}

We select the magnitude of the earthquakes as the mark of the pattern and we focus our analysis on the seismic sequence of L'Aquila that occurred in 2009.
Figure \ref{fig:appl4} (b) illustrates the Abruzzo region, containing the 2009 sequence of disruptive earthquakes. The event considered to be the mainshock occurred at 03:32 a.m. on the $6^{\text{th}}$ April 2009 with a magnitude of 5.9, and it is represented by the black point. Figure \ref{fig:appl4} (c-d) illustrates the temporal distribution of the seismic events, and the observed unmarked $K$-function.  From this, we know that the point pattern, consisting of 407 events, is highly clustered.   The spatial and temporal location of the mainshock is represented in black.

The global test based on $\mathcal{T}_1$ leads to the rejection of the null hypothesis, signifying that there is structure in both the points and the marks. This is further corroborated by the rejection of the null hypothesis of both $\mathrm{H}2$ and $\mathrm{H}3$.

Furthermore, the result of the local test $\mathrm{H}1_{L}$ provides 297  significant points, that is, the 73\% of the total. These local results, together with the rejection of the global hypotheses informs us that a marked inhomogeneous point process model should be employed to describe the observed point pattern, e.g. the ETAS model.

For each variable available in the catalog, that is, longitude, latitude, depth and time, the Kolmogorov–Smirnov test comparing the distributions of significant and non-significant points is run. This yielded p-values near zero for all variables except depth (see Table \ref{tab:ks_test}). This provides evidence against the null hypothesis that the distributions of longitude, latitude, and time are identical across the two subgroups.

\begin{table}[!htb]
\centering
\caption{Kolmogorov-Smirnov test results marginally comparing distributions of latitude, longitude, depth, and time between the two groups of significant and non-significant points according to the local test for hypothesis $\mathrm{H}1_{L}$. Significance is assessed at $\alpha = 0.05$.}
\label{tab:ks_test}

\begin{tabular}{lc}
\toprule
\textbf{Variable} & \textbf{p-value}  \\
\midrule
Longitude & $2.41 \times 10^{-3}$  \\
Latitude  & $6.77 \times 10^{-6}$  \\
Depth     & $1.81 \times 10^{-1}$  \\
Time      & $2.43 \times 10^{-3}$  \\
\bottomrule
\end{tabular}
\end{table}

Figure \ref{fig:resultsETAS} shows the distributions of the longitude, latitude, and time of the analysed earthquakes by their significance according to the local test. In red, the distributions of the significant points and in grey, the distributions of the non-significant ones. The temporal location of the mainshock is in black.

\begin{figure}[!htb]
	\centering
\subfloat{\includegraphics[width=.33\textwidth]{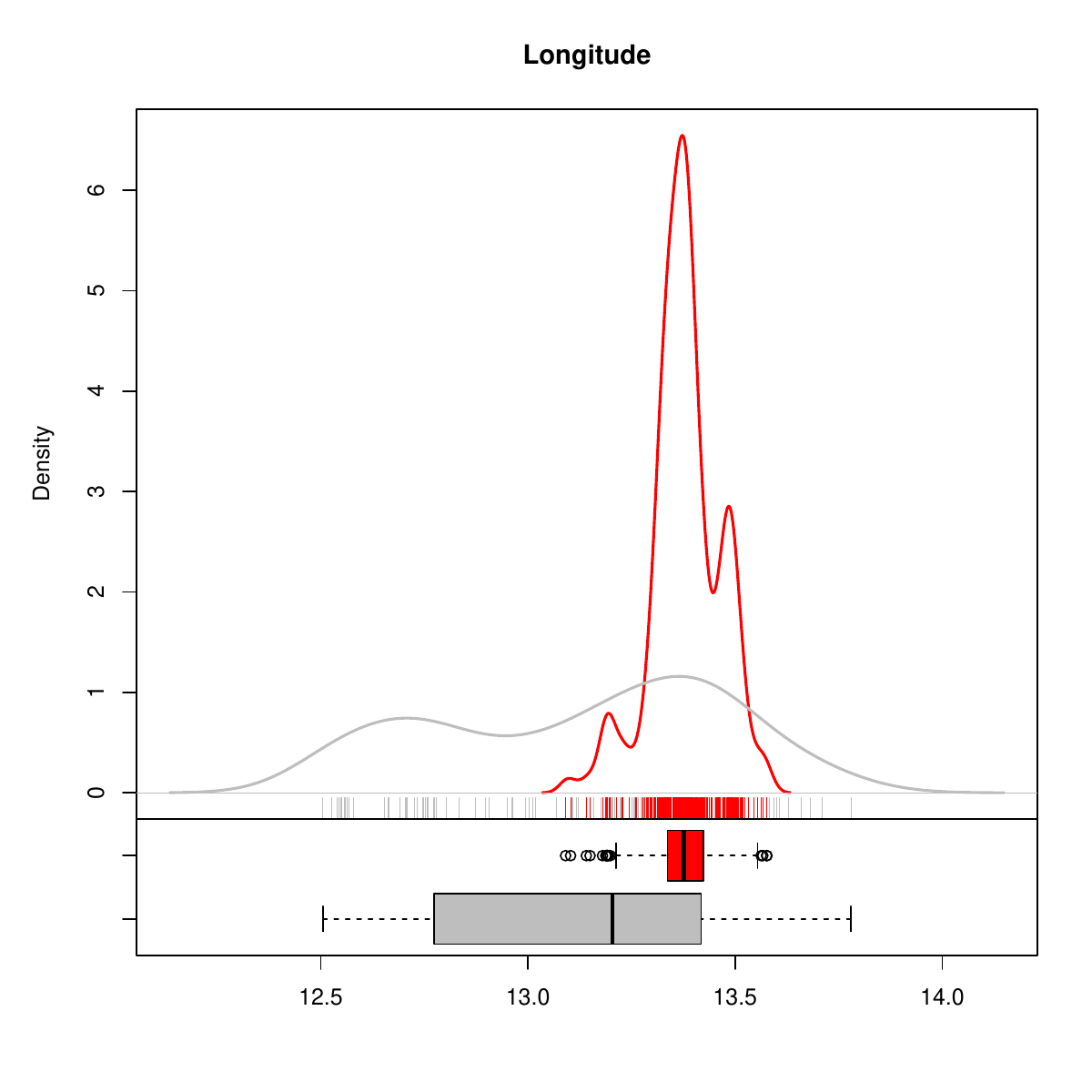}}  
\subfloat{\includegraphics[width=.33\textwidth]{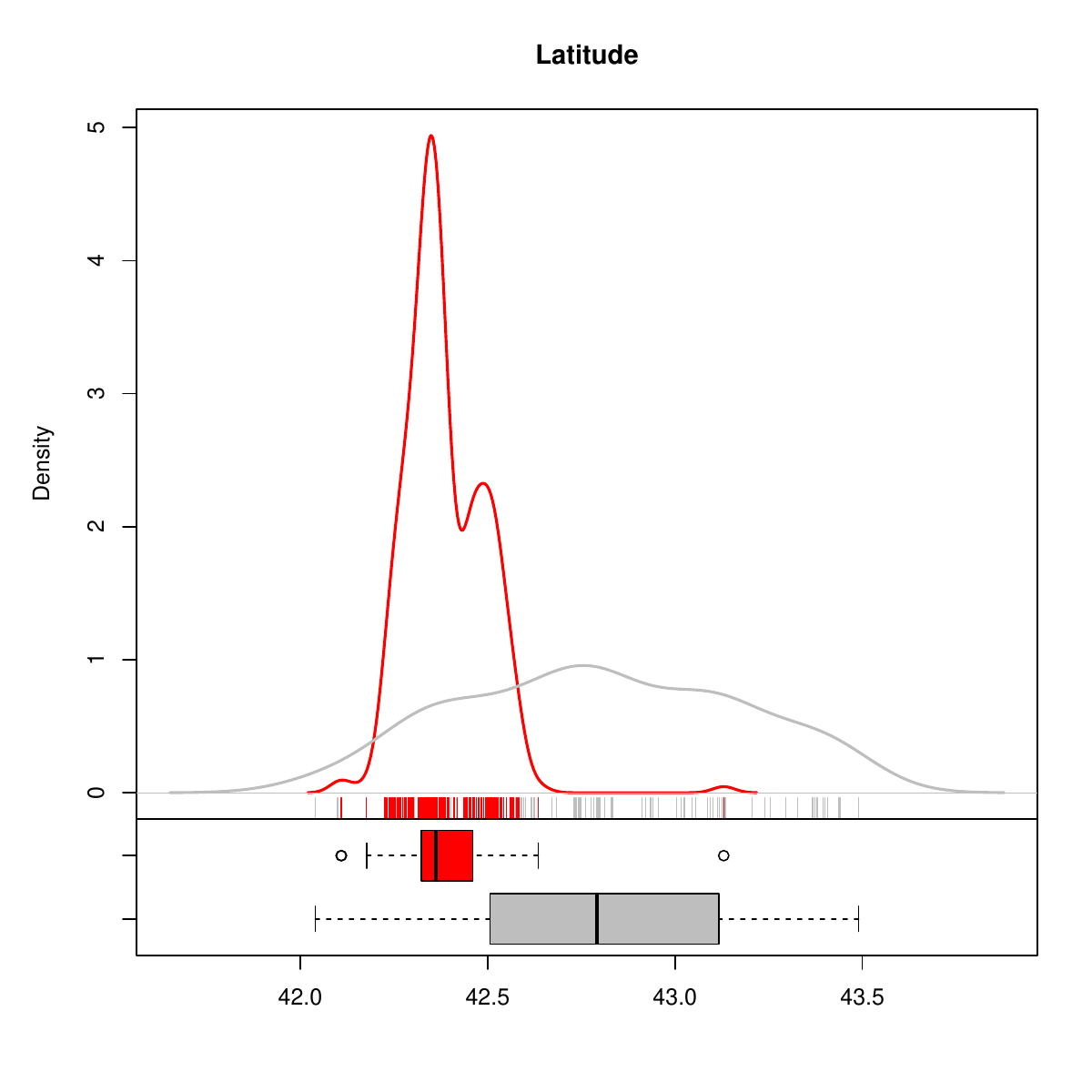}}
\subfloat{\includegraphics[width=.33\textwidth]{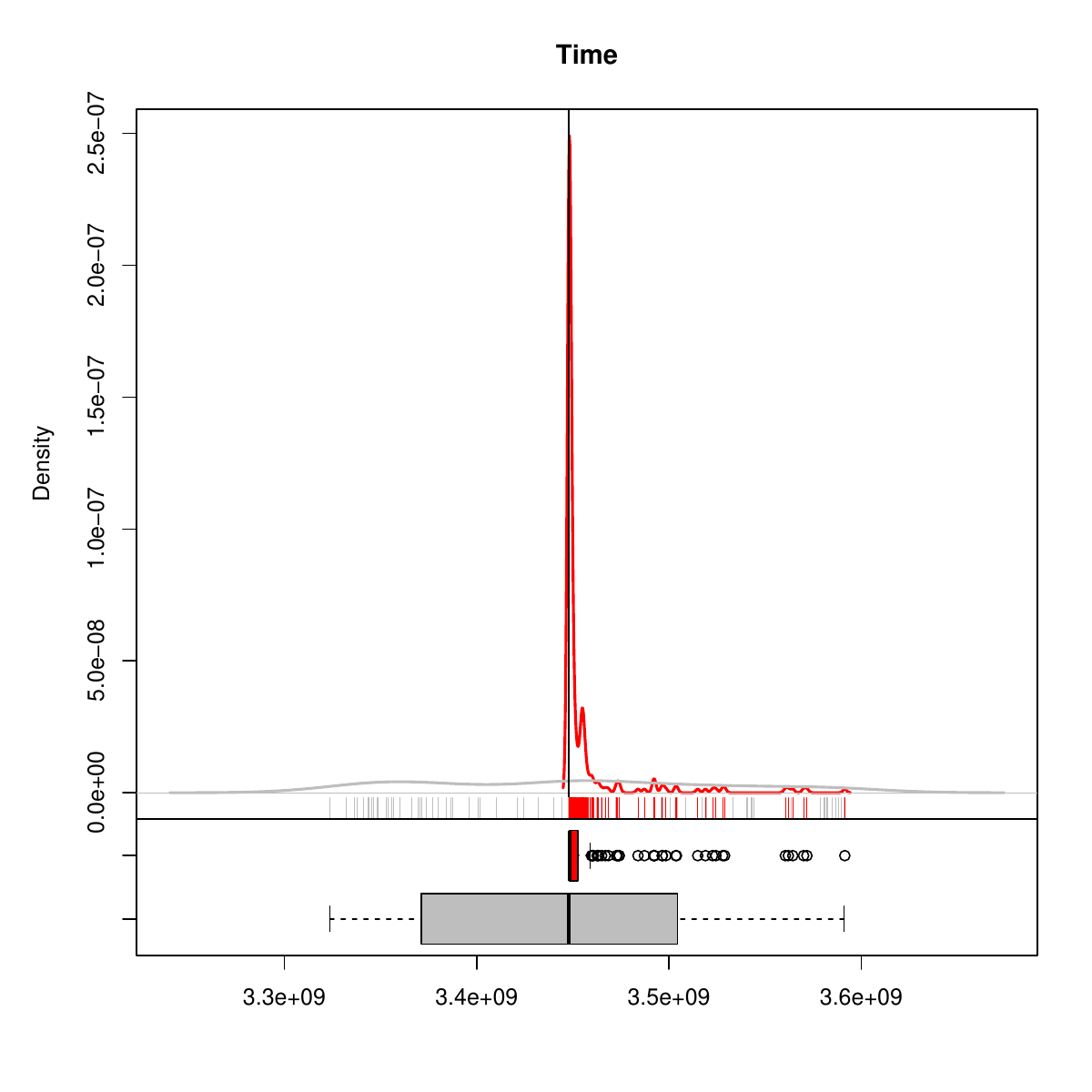}}  
	\caption{Distributions of the longitude, latitude, and time of the analysed earthquakes by their significance according to the local test. In red, the distributions of the significant points and in grey, the distributions of the non-significant ones. The temporal location of the mainshock is in black.
    }
	\label{fig:resultsETAS}
\end{figure}

Furthermore, the visual inspection of the spatio-temporal distribution of points by significance to the local test is illustrated in Figure \ref{fig:appl6}. 
In particular, panel (a) reports the spatial distribution, panel (b) shows the spatio-temporal distribution, and panel (c) depicts the depth as the third dimension. Here again, the red points indicate the significant ones while the grey ones are the non-significant ones. A spatial separation between the two groups of points is evident. The significant points correspond to seismic events that occurred mainly after the mainshock, and up to the end of the analysed time frame.

\begin{figure}[!htb]
	\centering
        \subfloat[2D]{\includegraphics[width=.5\textwidth]{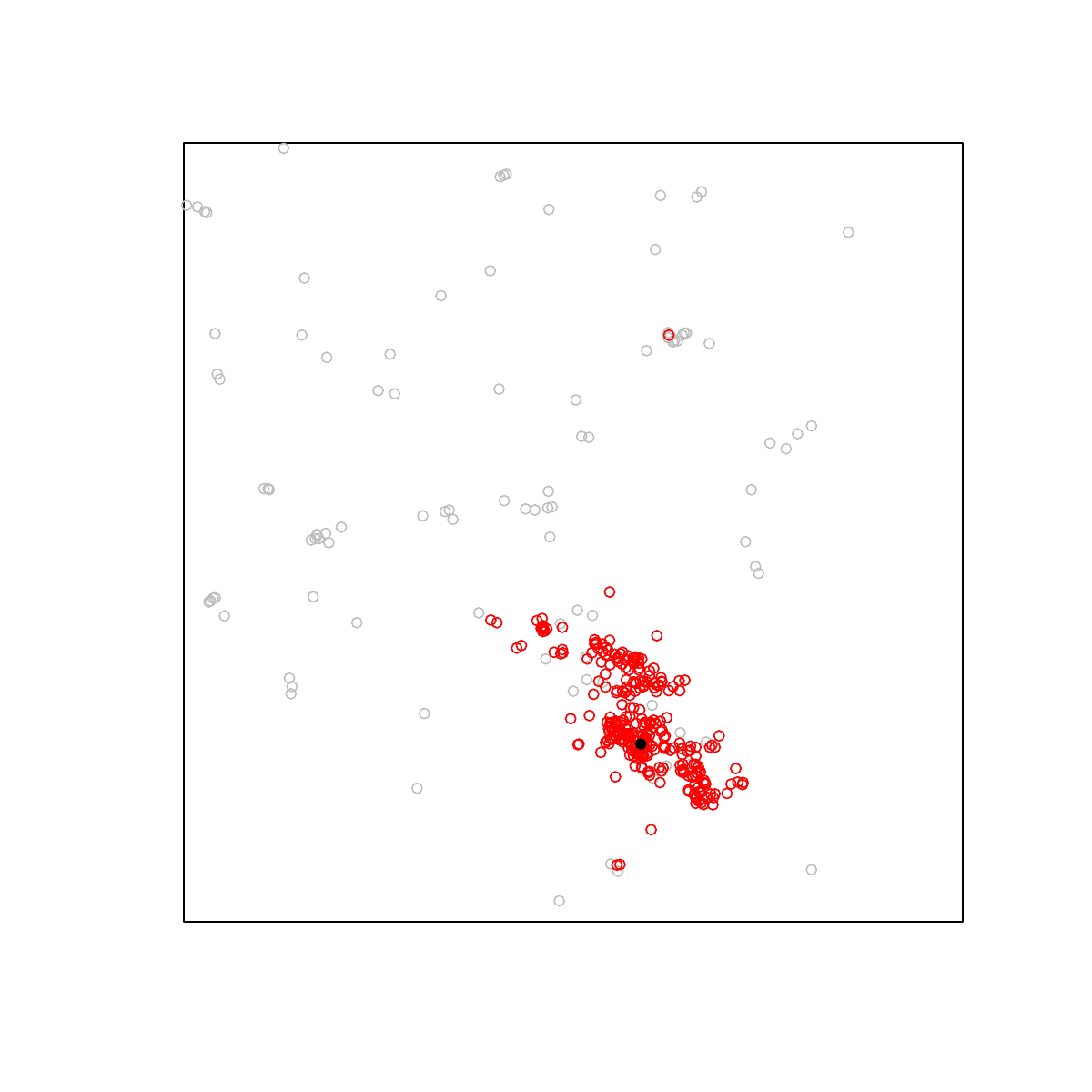}} 
\subfloat[3D time ]{\includegraphics[width=.5\textwidth]
{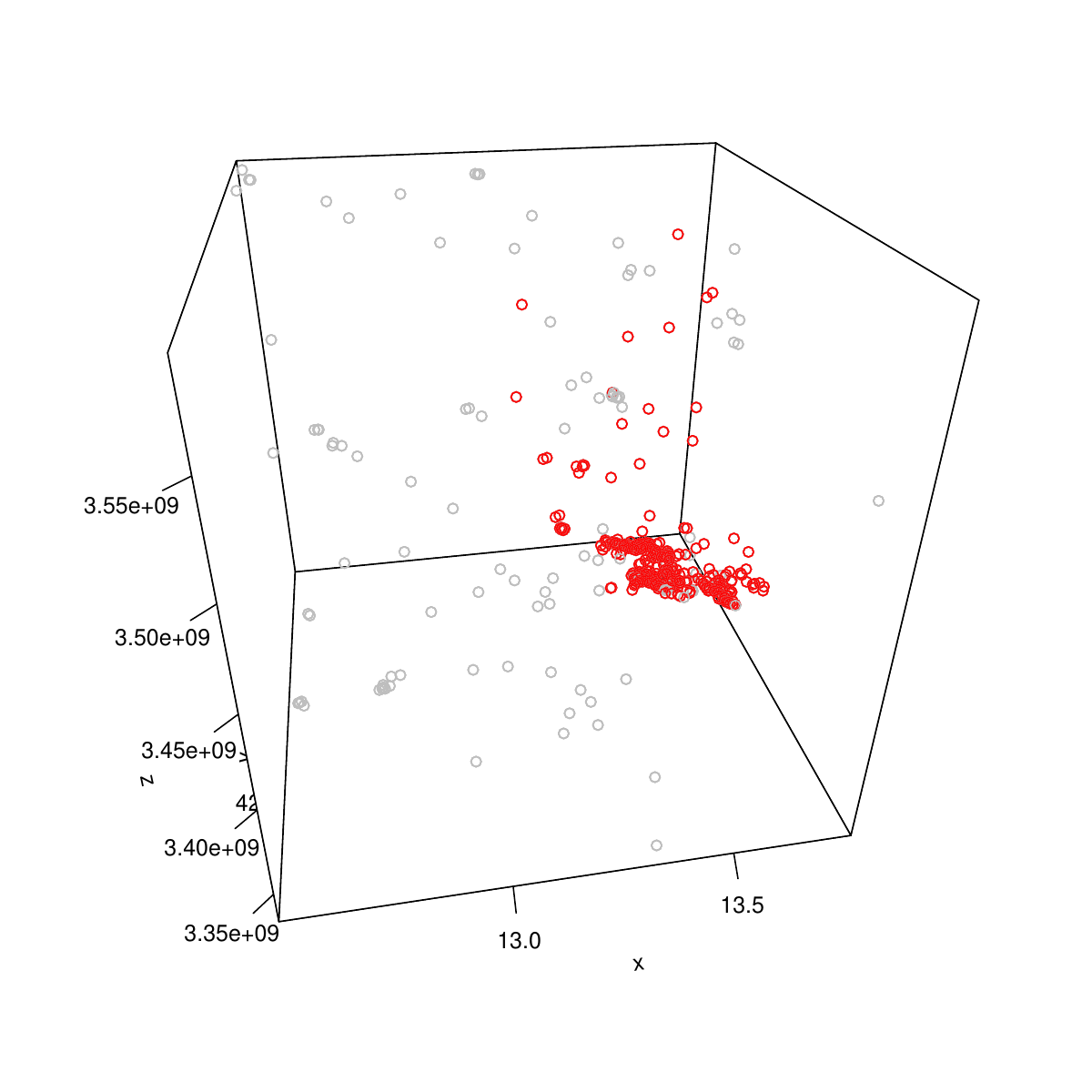}
} 
	\caption{Spatial (a) and spatio-temporal (b) distribution of the unmarked point pattern. In red, the significant points of the local test for hypothesis $\mathrm{H1}_{L}$. The spatial location of the mainshock is in black.}
	\label{fig:appl6}
\end{figure}

\section{Conclusions}\label{sec:concl}

This work introduces a comprehensive statistical framework for testing hypotheses on marked point patterns, based on global and local extensions of the mark-weighted inhomogeneous $K$-function. The proposed test statistics, formulated in the spirit of chi-squared discrepancies, enable the assessment of various hypotheses concerning the interaction between spatial locations and mark values, under both homogeneous and inhomogeneous settings.

The global test accounts for spatial inhomogeneity, effectively isolating mark-related effects from purely spatial ones. The local version offers deeper insight by identifying which points contribute most to deviations from the null hypothesis, thereby revealing localised spatial structures that would otherwise be obscured by aggregate global summaries. The simulation study demonstrates that the global test reliably discriminates among different configurations of point and mark structures, with power increasing as the number of points grows. The local test proves effective at detecting spatially localised structures, with well-controlled false positive rates and similarly improving power with sample size. Applications to real environmental datasets further confirm the practical applicability and interpretability of the proposed framework. In particular, analyses of forestry and earthquake data illustrate how the local test complements global conclusions and can guide spatially targeted modelling decisions. Several promising directions remain open for future work. First, the proposed statistics could be integrated into model-fitting strategies, for instance, the local mark-weighted $K$-function may serve as a regularisation term within penalised likelihood frameworks \citep{d2024constructed,tarantino2025modeling}, simultaneously capturing spatial and mark-related interactions without imposing restrictive assumptions on their relationship. Second, the methodology is naturally extensible to multivariate or functional marks, as well as to spatio-temporal settings. Third, the inhomogeneous mark-weighted local $K$-function, not yet fully explored here, offers potential as a diagnostic tool for fitted marked first-order intensity functions — particularly valuable in highly clustered or irregular patterns where mark dependence cannot be neglected, and where selecting an appropriate intensity estimate remains challenging.

\section*{Acknowledgement}

The authors gratefully acknowledge financial support through the German Research Association and the
ITALY. Nicoletta D'Angelo and Giada Adelfio were partially funded by the Targeted Research Funds 2026 (FFR 2026) of the University of Palermo, PRIN 2022: Spatio-temporal Functional Marked Point Processes for
probabilistic forecasting of earthquakes 2022BN7CJP P.I. Giada Adelfio. CUP B53C24006340006, and by Fondi di Ateneo – Bando EUROSTART 2025 – progetto di ricerca ASPIRE – Advances in Spatio-temporal Point process models Innovations and Research P.I. Nicoletta D'Angelo.

\bibliographystyle{chicago}
\bibliography{arXivAAA}
\end{document}